\documentclass[10pt,journal]{IEEEtran}

\usepackage[T1]{fontenc}% optional T1 font encoding

\ifCLASSOPTIONcompsoc

  \usepackage[nocompress]{cite}
\else

  \usepackage{cite}
\fi

\ifCLASSINFOpdf
\else
\fi

\usepackage{amsmath}
\makeatletter
\def\ps@headings{%
\def\@oddhead{\mbox{}\scriptsize\rightmark \hfil \thepage}%
\def\@evenhead{\scriptsize\thepage \hfil \leftmark\mbox{}}%
\def\@oddfoot{}%
\def\@evenfoot{}}
\makeatother 

\usepackage[colorlinks,
linkcolor=black,
anchorcolor=black,
citecolor=black,
urlcolor=black,
bookmarks=false
]{hyperref}

\usepackage{cite}
\usepackage{dsfont}

\usepackage{graphicx}
\graphicspath{{figs/}}
\usepackage{subfigure}
\usepackage{color}
\usepackage{ifpdf}
\usepackage{epsfig}
\usepackage{epstopdf}
\usepackage{latexsym}
\usepackage{amsfonts}
\usepackage{amssymb}
\usepackage{paralist}
\usepackage{comment}
\usepackage{xspace}
\usepackage{mathrsfs}
\usepackage{graphicx}
\usepackage{amssymb}
\usepackage{setspace}
\usepackage{caption}
\usepackage{algorithm}
\usepackage[noend]{algorithmic}

\usepackage{url,epsfig,array}
\usepackage{leftidx}
\usepackage{amsmath}

\usepackage{diagbox}
\usepackage{makecell}
\usepackage{multirow}
\usepackage{booktabs}
\usepackage{enumerate}
\usepackage{pifont}
\usepackage{threeparttable}

\usepackage{amsthm}

\usepackage{ragged2e}
\usepackage{balance}

\allowdisplaybreaks[4]

\usepackage{array}
\newcommand{\PreserveBackslash}[1]{\let\temp=\\#1\let\\=\temp}
\newcolumntype{C}[1]{>{\PreserveBackslash\centering}p{#1}}
\newcolumntype{R}[1]{>{\PreserveBackslash\raggedleft}p{#1}}
\newcolumntype{L}[1]{>{\PreserveBackslash\raggedright}p{#1}}

\def\ie{\textit{i.e.}\xspace}

\def\eg{\textit{e.g.}\xspace}

\def\st{\xspace\textbf{s.t.}\xspace}

\newtheoremstyle{colon}%
  {}{}%
  {\itshape}{}%
  {\bfseries}{:}%
  { }%
  {\thmname{#1}\thmnumber{ #2} \mdseries\thmnote{#3}}

\theoremstyle{colon}

\newtheorem{theorem}{Theorem}

\newtheorem{corollary}{Corollary}
\newtheorem{definition}{Definition}
\newtheorem{lemma}{Lemma}

\newtheorem{remark}{Remark}

\newtheorem{assumption}{Assumption}

\renewcommand{\algorithmicrequire}{\textbf{Input:}}
\renewcommand{\algorithmicensure}{\textbf{Output:}}

\begin{document}

\title{DySTop: Dynamic Staleness Control and Topology Construction for Asynchronous Decentralized Federated Learning}

\author{Yizhou Shi, Qianpiao Ma, Yan Xu, Junlong Zhou,~\IEEEmembership{Member,~IEEE}, 
Ming Hu,~\IEEEmembership{Member,~IEEE}, \\
Yunming Liao, Hongli Xu,~\IEEEmembership{Member,~IEEE}

\IEEEcompsocitemizethanks{\IEEEcompsocthanksitem Corresponding authors: Qianpiao Ma, Junlong Zhou.}

\IEEEcompsocitemizethanks{\IEEEcompsocthanksitem Yizhou Shi, Qianpiao Ma, Yan Xu and Junlong Zhou are with the School of Computer Science and Engineering, Nanjing University of Science and Technology, Nanjing, Jiangsu 210094, China (e-mail: yizhou\_shi@njust.edu.cn, maqianpiao@njust.edu.cn, 084621120@njucm.edu.cn, jlzhou@njust.edu.cn).\protect
}

\IEEEcompsocitemizethanks{\IEEEcompsocthanksitem Ming Hu is with the School of Computing and
Information Systems, Singapore Management University, Singapore (e-mail: hu.ming.work@gmail.com). \protect
}

\IEEEcompsocitemizethanks{\IEEEcompsocthanksitem Yunming Liao and Hongli Xu are with the School of Computer Science and Technology, University of Science and Technology of China, Hefei, Anhui 230027, China (e-mail: ymliao98@mail.ustc.edu.cn, xuhongli@ustc.edu.cn).\protect
}

}

\markboth{IEEE Transactions on Mobile Computing,~Vol., No., Jul.~2025}
{Shi \MakeLowercase{\textit{et al.}}: DySTop: Dynamic Staleness Control and Topology Construction for Asynchronous Decentralized Federated Learning}

\IEEEtitleabstractindextext{%
\begin{abstract}
\justifying
Federated Learning (FL) has emerged as a potential distributed learning paradigm that enables model training on edge devices (\textit{i.e.}, workers) while preserving data privacy. However, its reliance on a centralized server leads to limited scalability.
Decentralized federated learning (DFL) eliminates the dependency on a centralized server by enabling peer-to-peer model exchange.
Existing DFL mechanisms mainly employ synchronous communication, which may result in training inefficiencies under heterogeneous and dynamic edge environments.
Although a few recent asynchronous DFL (ADFL) mechanisms have been proposed to address these issues, they typically yield stale model aggregation and frequent model transmission, leading to degraded training performance on non-IID data and high communication overhead. To overcome these issues, we present DySTop, an innovative mechanism that jointly optimizes dynamic staleness control and topology construction in ADFL. In each round, multiple workers are activated, and a subset of their neighbors is selected to transmit models for aggregation, followed by local training. We provide a rigorous convergence analysis for DySTop, theoretically revealing the quantitative relationships between the convergence bound and key factors such as maximum staleness, activating frequency, and data distribution among workers. From the insights of the analysis, we propose a worker activation algorithm (WAA) for staleness control and a phase-aware topology construction algorithm (PTCA) to reduce communication overhead and handle data non-IID. Extensive evaluations through both large-scale simulations and real-world testbed experiments demonstrate that our DySTop reduces completion time by 51.8\% and the communication resource consumption by 57.1\% compared to state-of-the-art solutions, while maintaining the same model accuracy.

\end{abstract}

}

\maketitle

\IEEEdisplaynontitleabstractindextext

\IEEEpeerreviewmaketitle

\section{Introduction}\label{sec:introduction}

\IEEEPARstart{R}{ecent} advances in the Internet of Things (IoT) promote continuous generation of vast quantities of data generated from physical environments \cite{hu2024industrial,zhou2023swarm}. Generally, this data is sent to a remote cloud for processing, raising concerns about potential privacy leakage and considerable latency. In response, \textit{edge computing} has evolved as a viable solution by distributing computation capacity to edge devices (\ie, workers) for accelerated local data processing. Furthermore, it promotes the adoption of federated learning (FL), which coordinates workers to perform distributed machine learning \cite{mcmahan2017communication} while preserving data privacy.

Traditional FL, also known as Centralized FL (CFL), comprises a set of workers and a centralized parameter server (PS)\cite{mcmahan2017communication,wang2019adaptive,Li2020On}. These workers collaboratively train a global model by performing local updates on their own datasets and subsequently transmitting their model parameters to the PS. The PS next aggregates these uploaded models and distributes the aggregated model back to the workers for the subsequent training round. This procedure usually maintains multiple training rounds until model convergence. Though CFL is a mature and widely used scheme, it comes with notable limitations. First, its reliance on a star topology inherently restricts system scalability. Second, frequent communication between the PS and workers generates a substantial traffic workload and creates a communication bottleneck at the PS, posing a significant risk of a single point failure.

\begin{table*}[thbp]\centering
{
\caption{Comparison of different DFL mechanisms.}\label{tbl:FLcomparison}
\begin{tabular}{c|c|ccccc} %{.42\textwidth}
\hline
\multicolumn{2}{c|}{DFL Mechanisms} & \makecell{Handing Edge\\ Heterogeneity} & \makecell{Handing Edge\\ Dynamic} & \makecell{Communication\\ Consumption} & \makecell{Handling\\ Non-IID} & \makecell{Handling\\ Staleness}  \\ \hline
\multirow{5}{*}{Synchronous} & BrainTorrent \cite{roy2019braintorrent} & Poor & Poor & High & Poor & -  \\
 & GossipFL \cite{tang2022gossipfl} & Poor & Poor & High  & Poor & -  \\
 & D-Cliques \cite{bellet2022d} & Poor & Poor & Medium & Good & -  \\
 & MATCHA \cite{wang2019matcha} & Poor & Medium & Low & Poor & -  \\
 & L2PL \cite{xu2021decentralized} & Poor & Medium & Low & Poor & -  \\
 & FedHP \cite{liao2023adaptive} & Poor & Medium & Low & Good & -  \\  \hline
\multirow{5}{*}{Asynchronous} & HADFL \cite{cao2021hadfl} & Good & Poor & High & Poor & Poor  \\
 & AsyNG \cite{chen2023enhancing} & Good & Good & Medium & Good & Poor  \\
 & AsyDFL \cite{liao2024asynchronous} & Good & Good & Medium & Good & Poor  \\
 & SA-ADFL \cite{ma2024dynamic} & Good & Good & High & Poor & Medium  \\
 & \textbf{DySTop (Ours)} & Good & Good & Low & Good & Good  \\ \hline
\end{tabular}
}
\vspace{-0.6cm}
\end{table*}

Decentralized federated learning (DFL) has emerged as a solution to overcome the limitations of CFL \cite{roy2019braintorrent,tang2022gossipfl,wang2019matcha,xu2021decentralized,liao2023adaptive,bellet2022d}. In DFL, each worker performs local model updating and transmits the updated models to its neighbors through peer-to-peer (P2P) communications.
Since DFL does not rely on a centralized PS, it can achieve high scalability and construct network topology more flexibly. When a worker in the system fails, it is only necessary for the affected workers to re-establish P2P connections with other workers. Moreover, each worker communicates with its neighbors, effectively distributing the communication overhead that was originally borne solely by the PS in CFL. However, achieving high efficiency in DFL presents unique challenges that need to be addressed.

\begin{itemize}
\item \textbf{Edge Heterogeneity:} Each worker interacts with heterogeneous neighbors with varying computational and communication capabilities, data sizes, and physical locations\cite{cui2023hiera}. Therefore, the coordination required for training in DFL is more complex compared with CFL.
\item \textbf{Edge Dynamic:} The network topologies may change over time due to intermittent connections among workers caused by their mobility \cite{fu2024dta}. Furthermore, a worker with poor communication conditions may lead to disconnections for its neighbors.
\item \textbf{Communication Resource Constraint:} On one hand, since the absence of the PS, DFL typically requires more training rounds to converge than CFL. On the other hand, in DFL, each worker sends its model to multiple workers rather than a single PS, resulting in more communication resource consumption compared to CFL \cite{tang2022gossipfl}.
\item \textbf{Data Non-IID:} The data collected on workers is often not a uniform sample of the overall distribution, leading to data non-independent-and-identically-distributed (non-IID) \cite{mcmahan2017communication}, which degrades FL training. Even more critically, unlike the PS in CFL, which trains the global model over all data, each worker in DFL may obtain models trained over more biased data \cite{chen2023enhancing}.
\end{itemize}

There are two communication schemes for DFL: synchronous and asynchronous. Most of the existing DFL researches adopt synchronous communication \cite{roy2019braintorrent,tang2022gossipfl,bellet2022d,wang2019matcha,xu2021decentralized,liao2023adaptive}, where each worker waits for all its neighbors to finish local model updating and aggregates them with its own local model. However, due to edge heterogeneity, workers are compelled to wait for the slowest worker to complete its local updating and model transferring, known as the \textit{straggler problem}.
Moreover, edge dynamics may render synchronous DFL impractical, as dynamic network conditions and potential disconnections can lead to significant idle time while workers wait for model transmissions from all neighbors in each round.

To handle heterogeneity and dynamics in edge networks, the asynchronous DFL (ADFL) scheme\cite{cao2021hadfl,chen2023enhancing,liao2024asynchronous,ma2024dynamic} further eliminates the synchronization barrier in synchronous DFL. In ADFL, each worker is allowed to exchange models with its neighbors at any time and aggregate models it has received so far upon finishing previous local updating.
Since a worker in ADFL does not need to wait for all its neighbors to transmit their latest local models, ADFL effectively addresses edge heterogeneity and edge dynamics. However, the existence of stragglers still causes different relative frequencies of workers to perform asynchronous updating, which leads to \textit{staleness} concern \cite{xie2019asynchronous}. Consequently, workers may aggregate the out-of-date models from their neighbors, causing unsatisfactory training performance\cite{ma2021fedsa}.
To this end, our previous work \cite{ma2024dynamic} SA-ADFL realizes dynamic staleness control for ADFL by appropriately determining a worker to perform model aggregation and transferring in each round. However, in SA-ADFL, the determined worker sends its model to all neighbors within its communication range, resulting in significant communication and lacking fine-grained handling of data non-IID.

To address this, we propose DySTop, a novel mechanism that performs \textbf{Dy}namic \textbf{S}taleness Control and \textbf{Top}ology Construction for ADFL. In each round, DySTop dynamically activates multiple workers, each selecting a subset of its neighbors and pulling their models for aggregation. We provide a rigorous convergence analysis for DySTop, revealing the significance of both staleness control and topology construction in achieving satisfactory training performance. Guided by the analysis, we design a worker activation algorithm to achieve more effective staleness control than SA-ADFL, and a phase-aware topology construction algorithm considering both communication efficiency and data non-IID. The summarization of different DFL mechanisms are shown in Table \ref{tbl:FLcomparison}.

Our main contributions to this paper are as follows:
\begin{itemize}
    \item We propose DySTop, a dynamic staleness control and topology construction mechanism for ADFL, and provide its convergence analysis, theoretically revealing how the convergence bound depends on factors such as maximum staleness, activating frequency, and data distribution.
     \item Based on the convergence analysis, we formulate a training time minimization problem for DySTop given staleness and communication resource constraints, and transform the problem into multiple per-round subproblems by Lyapunov optimization. 
    \item To solve the subproblems, we design a worker activation algorithm (WAA) for staleness control and a phase-aware topology construction algorithm (PTCA) for reducing communication overhead and handling data non-IID.
    \item We conduct both large-scale simulations and real-world testbed experiments to evaluate the performance of our proposed mechanism and algorithms. Experimental results show that our algorithms can reduce completion time by 51.8\% and the communication resource consumption by 57.1\% under the same accuracy requirement compared with the state-of-the-art mechanisms.
\end{itemize}

The rest of this paper is organized as follows. Section \ref{sec:related} reviews the related works. Section \ref{sec:prelim} introduces the DySTop mechanism and formulates the problem. Section \ref{sec:convergence} provides the convergence analysis. The algorithms for DySTop are proposed in Section \ref{sec:algorithm}. The experimental results of the simulations and testbeds are shown in Sections \ref{sec:evaluation} and \ref{sec:implementation}, respectively. Section \ref{sec:conclusion} concludes this paper.

\section{Related Works}\label{sec:related}

The communication schemes of DFL determine how workers perform local training and exchange models, and can be categorized as either synchronous or asynchronous. Most existing DFL studies adopt synchronous communication schemes, which are reviewed in Section \ref{subsec:reviewSDFL}. More recently, a few DFL studies have explored asynchronous schemes, as reviewed in Section \ref{subsec:reviewADFL}. Although asynchronous schemes can improve resource utilization and reduce idle time, they introduce model staleness. The techniques for mitigating staleness in asynchronous FL are reviewed in Section \ref{subsec:stalenesscontrol}.

\subsection{Communication Schemes of Synchronous DFL}\label{subsec:reviewSDFL}

In synchronous DFL, each worker waits for all its neighbors to finish local model updating and performs model aggregation. Therefore, it is important to construct an efficient topology for model parameter transmission in DFL \cite{palmieri2024impact}.
BrainTorrent \cite{roy2019braintorrent} establishes a fully-connected DFL, requiring each worker to share its local model with all others in each round. However, it incurs heavy communication costs due to the large volume of data transmissions.
Partially connected topologies reduce the communication overhead by enabling each worker to exchange its local model only with its neighboring workers. For example, GossipFL \cite{tang2022gossipfl} creates a highly sparsified topology, allowing each worker to transmit its local model to only one peer in each round, realizing efficient communication. D-Cliques \cite{bellet2022d} proposes a new topology where workers are grouped into connected cliques based on data distribution, which helps reduce gradient bias. Yet, previous works tend to construct a fixed topology and ignore the edge dynamic. Therefore, some works focus on re-constructing the topology at regular intervals. For example, MATCHA \cite{wang2019matcha} employs matching decomposition to break down the base topology into disjoint subgraphs, and samples a subset of these subgraphs to construct a sparse topology in each round. L2PL \cite{xu2021decentralized} proposes a reinforcement learning-based method to dynamically construct an optimal sparse topology for each round. FedHP \cite{liao2023adaptive} controls local updating frequency and network topology adaptively to improve the overall learning efficiency.

However, the above synchronous DFL mechanisms share a common drawback: they force each worker to await the completion of all other workers, \ie, the synchronization point, before the subsequent round can begin. In edge heterogeneity environments, this process results in slow convergence, since the overall efficiency is dictated by the slowest worker.

\subsection{Communication Schemes of Asynchronous DFL}\label{subsec:reviewADFL}

A few studies have explored asynchronous communication schemes for DFL, allowing workers to exchange models whenever they finish local updating, instead of waiting for the synchronization points. For example, HADFL \cite{cao2021hadfl} facilitates decentralized training on heterogeneous workers in an asynchronous manner, with probabilistic participation in model synchronization and aggregation in each round. AsyNG \cite{chen2023enhancing} and AsyDFL \cite{liao2024asynchronous} develop a joint optimization framework integrating gradient pushing and neighbor selection to achieve a communication-training trade-off in ADFL.

However, these ADFL studies overlook the staleness problem inherent in asynchronous communication. This oversight results in highly stale models, which generate significant deviation from current gradient trajectories and intensify the detrimental effects of non-IID data.

\subsection{Deal with Staleness in Asynchronous FL}\label{subsec:stalenesscontrol}

As previously stated, staleness significantly impacts the performance of asynchronous FL. Many asynchronous CFL studies have focused on mitigating the negative effects of staleness. For example, FedAsync \cite{xie2019asynchronous} relieves the impact of staleness on training by assigning smaller aggregation weights to stale models. FedSA \cite{ma2021fedsa} employs a semi-asynchronous mechanism where partial workers participate in each global updating to control staleness. SAFA \cite{wu2020safa} tackles staleness by simply discarding models that become excessively stale during the training process. SC-AFL \cite{sun2024staleness} restricts the staleness degree within a certain bound by dynamically adjusting the aggregated strategy in each round. ASO-Fed \cite{chen2020asynchronous} deploys dynamic learning rates for workers based on their participation frequencies in global updating. Air-FedGA \cite{ma2025air} assigns workers into groups to mitigate staleness. The authors in \cite{zheng2017asynchronous} and \cite{zhu2022client} compensate for the gradients of stale models by leveraging the approximate Taylor expansion.

However, the above staleness control methods primarily focus on asynchronous CFL scenarios, whereas staleness control in ADFL is more complex and remains largely unexplored.
Our previous work SA-ADFL \cite{ma2024dynamic} controls staleness for ADFL by determining a worker to perform model updating and send its model to all neighbors in each round. However, it incurs high communication overhead and lacks fine-grained handling of data non-IID. In this paper, we propose DySTop, a mechanism that dynamically activates multiple workers, each selecting a subset of its neighbors and pulling their models for aggregation, thereby reducing communication overhead and effectively addressing data non-IID.

\iffalse
\subsection{Topology Construction of DFL}

Palmieri et al. \cite{palmieri2024impact} explores how different types of network topologies influence the spreading of knowledge (\textit{i.e.}, model parameter), which demonstrates the importance of the topological structure in DFL. Various topology construction methods \cite{roy2019braintorrent}, \cite{tang2022gossipfl}, \cite{wang2019matcha}, \cite{bellet2022d} have been proposed for faster model convergence. Yet, previous works tend to construct a fixed topology and ignore the dynamic and communication constraints in the edge network. To adapt to edge dynamic, recent works focus on re-constructing the topology at regular intervals.
L2PL \cite{xu2021decentralized} develops a learning-driven method to dynamically build an optimal partially connected topology in each training round. FedHP \cite{liao2023adaptive} adaptively controls local updating frequency and network topology to support the heterogeneous workers. Wei et al. \cite{wei2024joint} design a heuristic algorithm, selecting three topologies (centralized FL, decentralized FL, and hierarchical FL) to adapt to different resource constraints and model types.
However, none of the above methods provide guidance on topology construction from the perspective of non-IID data, which is also a significant factor influencing model performance.

\fi

\section{System Overview and Problem Formulation}\label{sec:prelim}

\subsection{Preliminary for ADFL}
We consider an asynchronous decentralized federated learning system,
where a set of workers $\mathcal{V}=\{v_1, v_2, \cdots, v_N\}$ is deployed to train local models.
Given the inherent edge dynamics that induce topological changes, we characterize the time-varying network topology through a directed graph $\mathcal{G}_t=\left(\mathcal{V}_t, \mathcal{E}_t\right)$ per round $t$. The vertex set $\mathcal{V}_t$ tracks available workers,  and the edge set $\mathcal{E}_t$ contains the links between devices, with $e(v_i, v_j)\in \mathcal{E}_t$ specifying the transfer path from $v_i$ to $v_j$.
For worker $v_i$, we define $\mathcal{N}_t^i=\{v_j | e(v_j, v_i)\in \mathcal{E}_t, v_j \in \mathcal{V}_t\}\cup\{v_i\}$ as its in-neighbors set that consists of all workers with an in-coming link to $v_i$ include itself.
Correspondingly, the out-neighbors set of $v_i$ at round $t$ is defined as $\hat{\mathcal{N}}_t^i=\{v_j | e(v_i, v_j)\in \mathcal{E}_t, v_j \in \mathcal{V}_t\}\cup\{v_i\}$.

Each worker $v_i \in \mathcal{V}$ conducts local training on its local dataset $\mathcal{D}_{i}=\bigcup_{k=1}^{|\mathcal{D}_{i}|}(\textbf{X}_{i}^k, y_{i}^k)$ with size $D_{i}=|\mathcal{D}_i|$, where $\textbf{X}_{i}^k$ denotes the feature vector of $k$-th sample and $y_{i}^k$ represents the corresponding label. Let $\mathbf{w}$ be the model parameter. The local loss function of worker $v_i$ is given by
\begin{equation}
F_i(\mathbf{w})\triangleq\mathbb{E}_{\xi_i\in\mathcal{D}_i}f(\mathbf{w};\xi_i)\mbox{,}
\end{equation}
where $f(\mathbf{w};\xi_i)$ computes the loss on mini-batch over mini-batch $\xi_i$ drawn from local dataset $\mathcal{D}_i$. $f(\cdot;\cdot)$ accommodates various loss formulations, \eg, cross-entropy, focal or hinge loss.
Define $\alpha_i=\frac{D_i}{D}$ as the relative data size of $v_i$ to the total, the global loss function is defined as
\begin{equation}\label{eq:loss}
F(\mathbf{w})\triangleq\sum_{v_i\in \mathcal{V}}\frac{D_i}{D}F_i(\mathbf{w})=\sum_{v_i\in \mathcal{V}}\alpha_iF_i(\mathbf{w})\mbox{.}
\end{equation}
The objective is to find the optimal parameter vector $\mathbf{w}^*$ so as to minimize $F(\mathbf{w})$, \textit{i.e.}, $\mathbf{w}^*=\text{argmin}_{\mathbf{w}}F(\mathbf{w})$.

Inspired by \cite{J1986round}, we divide the ADFL process into multiple rounds, and we introduce $t\in \mathcal{R}=\{1,2, \cdots, T\}$ to indicate the index of a round. At each round $t$, the local model of worker $v_i$ is represented as $\mathbf{w}_t^i$.
Due to the asynchronous updating mechanism and heterogeneous computational capabilities across workers, the local model updates often span multiple rounds. Let $\tau_t^i$ denote the interval (called the \textit{staleness}) between the latest starting-to-training round of $v_i$ and the current round $t$. The local model at round $t$ actually satisfies
\begin{equation}\label{eq:asyw}
\mathbf{w}_t^i=\mathbf{w}_{t-\tau_t^i}^i\mbox{.}
\end{equation}

\subsection{System overview of DySTop}

We propose a dynamic staleness control and topology construction mechanism for ADFL, called DySTop, which is formally described in Alg. \ref{alg:dystop}.
DySTop mainly involves two key roles, \textit{i.e.}, a set of workers and a coordinator. Periodically, the coordinator acquires worker-specific status data such as available bandwidth, model staleness, and data distribution, which it then uses to determine worker activation and topology construction strategies.
Now we introduce the process of DySTop on both the worker side and the coordinator side.

On the worker side (Lines \ref{alg:line:devicestart}-\ref{alg:line:deviceend}), each worker maintains two threads: an updating thread (Lines \ref{alg:line:updatingstart}-\ref{alg:line:updatingend}) and a pushing thread (Lines \ref{alg:line:pushingstart}-\ref{alg:line:pushingend}).
For updating thread, if worker $v_i$ receives an EXECUTE message from the coordinator at round $t$ and has completed its local training of the current round, it sends PULL message to its in-neighbor workers $v_j, \forall v_j\in\mathcal{N}_t^i$, then each in-neighbors $v_j$ transmits its local model $\mathbf{w}_t^j$ to $v_i$.
Next, $v_i$ aggregates the pulled local models by
\begin{equation}\label{eq:aggregation}
    \hat{\mathbf{w}}_t^i = \sum\nolimits_{v_j \in \mathcal{N}_t^i}\sigma_t^{i,j}\cdot\mathbf{w}_t^j= \sum\nolimits_{v_j \in \mathcal{N}_t^i}\sigma_t^{i,j}\cdot\mathbf{w}_{t-\tau_t^j}^j.
\end{equation}
where $\sigma_t^{i,j}$ is the aggregated weight of $v_j$ on the worker $v_i$ at round $t$, calculated as $\sigma_t^{i,j} = D_j / \sum_{v_{j'} \in \mathcal{N}_t^i}D_{j'}$.
After model aggregation, the worker $v_i$ performs local training as
\begin{equation}\label{eq:updating}
    \mathbf{w}_{t+1}^i = \hat{\mathbf{w}}_t^i - \eta \nabla F_i(\hat{\mathbf{w}}_t^i; \xi_t^i).
\end{equation}
where $\eta$ is the learning rate, and $\nabla F_i(\hat{\mathbf{w}}_t^i; \xi_t^i)$ denotes the stochastic gradient over mini-batch $\xi_t^i$.
Meanwhile, at any time receiving PULL message from the out-neighbor $v_j\in \hat{\mathcal{N}}_t^i$, the pushing thread of worker $v_i$ pushes its up-to-date local model, which is formulated in Eq.~\eqref{eq:asyw}.

On the coordinator side (Lines \ref{alg:line:coorstart}-\ref{alg:line:coorend}), at the beginning of each round $t$, the coordinator determines a subset of workers $\mathcal{A}_t\subseteq\mathcal{V}$ (we call the active set) for aggregation at round $t$, which will be elaborated in Alg. \ref{alg:device}. Each worker $v_i$ is assigned a binary indicator $a_t^i \in \{0, 1\}$. That is, $a_t^i=1$ if $v_i$ is activated for model aggregation at round $t$, and $a_t^i=0$ otherwise.
Then, the coordinator constructs the network topology $\mathcal{G}_t$, which will be elaborated in Alg. \ref{alg:topology}.
Next, the coordinator sends EXECUTE message to $v_i\in\mathcal{A}_t$, and updates staleness of each worker $v_i\in\mathcal{V}$ as
\begin{equation}\label{eq:staleness}
    \tau_{t+1}^i = \left(\tau_t^i+1\right)(1-a_t^i).
\end{equation}
That is, for worker $v_i\in \mathcal{A}_t$, the next round staleness $\tau_{t+1}^i$ is reset to 0, while the staleness of all other workers increases.

\begin{algorithm}[t]
\caption{Dynamic Staleness Control and Topology Construction for ADFL (DySTop)}\label{alg:dystop}
\begin{algorithmic}[1]
\FOR {$t=1$ to $T$}
\STATE \textbf{Processing at Each Device $v_i$:}
\STATE {\textit{\textbf{Updating Thread:}}}\label{alg:line:devicestart}
\IF {Receive EXECUTE message \textbf{and} finish local training}\label{alg:line:updatingstart}
\STATE {Obtain $\mathbf{w}_t^j$ of each device $v_{j}\in\mathcal{N}_t^i$ by sending PULL message to its in-neighbors;}
\STATE {Aggregate local models to obtain $\hat{\mathbf{w}}_t^i$ by Eq. \eqref{eq:aggregation};}
\STATE {Update local model $\mathbf{w}_{t+1}^i$ by Eq. \eqref{eq:updating};}
\ENDIF\label{alg:line:updatingend}
\STATE {\textit{\textbf{Pushing Thread:}}}\label{alg:line:pushingstart}
\IF {Receive PULL message from $v_j\in \hat{\mathcal{N}}_t^i$}
\STATE {Push local model $\mathbf{w}_{t-\tau_t^i}^i$ to $v_j$;}\label{alg:line:deviceend}
\ENDIF\label{alg:line:pushingend}
\STATE \textbf{Processing at the Coordinator:}\label{alg:line:coorstart}
\STATE {Determine active set $\mathcal{A}_t$ according to Alg. \ref{alg:device} to pull models and perform aggregation;}
\STATE {Contruct network topology $\mathcal{G}_t$ according to Alg. \ref{alg:topology};}
\STATE {Send EXECUTE message to $v_i \in \mathcal{A}_t$;}
\FOR {\textbf{each} $v_i\in\mathcal{V}$}
\STATE {Update staleness $\tau_{t+1}^i$ by Eq. \eqref{eq:staleness};}
\ENDFOR\label{alg:line:schedulingend}
\ENDFOR\label{alg:line:coorend}
\end{algorithmic}
\end{algorithm}

\subsection{Problem Formulation}

For DySTop, one problem is how to determine a proper active set $\mathcal{A}_t$ at each round $t$, resulting in the process of model training and transmission.
Let $H_t$ denote the duration of round $t$.
$h_{i}$ represents the local training time required for worker $v_i$. Then the time consumed by $v_i$ for local training at round $t$ is calculated as
\begin{equation}\label{eq:h_comp}
    h_{t}^{i,\rm cmp}= \max \left\{h_{i} - \sum_{k=t-\tau_t^i}^{t-1}H_k, 0\right\}.
\end{equation}
Let $h_t^{i,j,\rm com}$ denote the time consumed by the transmission of the model from $v_j$ to $v_i$ at round $t$.
Thus, the sum of training and transmission time of the activated worker $v_i$ at round $t$ is
\begin{equation}\label{eq:costtime}
    H_t^i = h_{t}^{i,\rm cmp}+\max_{v_j\in\mathcal{N}_t^i} \left\{h_t^{i,j,\rm com}\right\}\mbox{.}
\end{equation}
Therefore, the duration of round $t$ is calculated as
\begin{equation}\label{eq:duration}
    H_t = \max_{v_i\in\mathcal{A}_t} \left\{H_t^i\right\}.
\end{equation}

The other problem is how to construct a network topology $\mathcal{G}_t$ for model transmission given communication resource constraints.
Specifically, let $c_t^{i,j}\in\{0,1\}$ indicate whether the directed link from $v_j$ to $v_i$ is connected for model transmission at round $t$. $c_t^{i,j}=1$ if the link is connected, and otherwise $c_t^{i,j}=0$. Let $b$ denote the bandwidth resource consumption for transmitting one local model. Since both pulling models and pushing models consume bandwidth, the bandwidth consumption of worker $v_i$ at round $t$ is
\begin{equation}
B_t^i=\left(\sum\nolimits_{v_j\in\mathcal{N}_t^i}c_t^{i,j} + \sum\nolimits_{v_k\in\hat{\mathcal{N}}_t^i}c_t^{k,i}\right)\cdot b.
\end{equation}

Unlike centralized FL, which always maintains a global model in the PS, DFL has no real global model due to its fully decentralized setting.
Therefore, we define the global weighted model at round $t$ as
\begin{equation}\label{eq:weightedmodel}
\mathbf{w}_t=\sum_{v_i\in\mathcal{V}}\frac{D_i}{D}\mathbf{w}_t^i=\sum_{v_i\in\mathcal{V}}\alpha_i\mathbf{w}_t^i\mbox{.}
\end{equation}

Let $\boldsymbol{a} = \{\boldsymbol{a}_t|t\in[1,T]\}$ and $\boldsymbol{c} = \{\boldsymbol{c}_t|t\in[1,T]\}$ denote the worker activation and topology construction strategies employed throughout the $T$ rounds of training process, respectively, where $\boldsymbol{a}_t=\{a_t^i|v_i\in\mathcal{V}\}$ and $\boldsymbol{c}_t=\{c_t^{i,j}|v_i,v_j\in\mathcal{V}\}$. Our problem is formulated as
\begin{subequations}
\begin{align}
\textbf{(P1)}:&\min_{\boldsymbol{a}, \boldsymbol{c}}\sum_{t=1}^{T}H_t \\ \label{cons:loss}
{\st} \quad
&F(\mathbf{w}_T) \leq \epsilon,  \\ \label{cons:stale}
&\tau_t^i\leq\tau_{\text{bound}}, &\forall v_i\in\mathcal{V}, t\in\mathcal{R}, \\ \label{cons:bw}
&B_t^i \leq \hat{B}_t^i, &\forall v_i\in\mathcal{V}, t\in\mathcal{R}, \\ \label{cons:sch}
&a_t^i\in\{0, 1\}, &\forall v_i\in\mathcal{V}, t\in\mathcal{R}, \\ \label{cons:topo}
&c_t^{i,j}\in\{0, 1\}, &\forall v_i, v_j\in\mathcal{V}, t\in\mathcal{R}
\mbox{.}
\end{align}
\end{subequations}
The inequality \eqref{cons:loss} ensures the convergence of global loss function after $T$ rounds, where $\epsilon$ is the convergence threshold. The set of inequalities \eqref{cons:stale} represents that the staleness of each worker cannot exceed the maximal bound $\tau_{\text{bound}}$. The set of inequalities \eqref{cons:bw} represents that bandwidth allocation for each worker should not exceed its bandwidth resource budget $\hat{B}_t^i$, which is time-varying considering the dynamic network condition. Our goal is to jointly optimize the worker activation strategy $\boldsymbol{a}$ and topology construction strategy $\boldsymbol{c}$ to minimize the overall training duration in DySTop, \ie, $\min_{\boldsymbol{a}, \boldsymbol{c}}\sum_{t=1}^{T}H_t$.

\section{Convergence Analysis}\label{sec:convergence}

\subsection{Assumptions}\label{subsec:assumptions}

We use the following assumptions that are classical to the analysis of the loss functions $F_i(\mathbf{w}), \forall v_i\in\mathcal{V}$ \cite{liao2024asynchronous}\cite{wang2019adaptive}\cite{ma2024feduc}.
\begin{assumption}[(Smoothness)]\label{ass:smoothness}
$F_i(\mathbf{w})$ is $L$-smooth, \ie, $\forall \mathbf{w}_1, \mathbf{w}_2$, $F_i(\mathbf{w}_2)-F_i(\mathbf{w}_1)\le\langle \nabla F_i(\mathbf{w}_1),\mathbf{w}_2-\mathbf{w}_1\rangle+\frac{L}{2}{\|\mathbf{w}_2-\mathbf{w}_1\|}^2$.
\end{assumption}
\begin{assumption}[(Convexity)]\label{ass:strongconvexity}
$F_i(\mathbf{w})$ is $\mu$-strongly convex, \ie, $\forall \mathbf{w}_1, \mathbf{w}_2$, $F_i(\mathbf{w}_2)-F_i(\mathbf{w}_1)\ge\langle \nabla F_i(\mathbf{w}_1), \mathbf{w}_2-\mathbf{w}_1 \rangle+\frac{\mu}{2}{\|\mathbf{w}_2-\mathbf{w}_1\|}^2$.
\end{assumption}
Note that Assumption \ref{ass:strongconvexity} holds for those models with convex loss functions, like support vector machines and linear regression. Our mechanism can also work for models (\eg, CNN) with non-convex loss functions, which is shown in Section \ref{sec:evaluation}.

\subsection{Analysis of Convergence Bounds}\label{subsec:convergence}

To describe the data distribution among workers, we give the following definitions, which are widely adopted in the existing works \cite{wang2019adaptive}\cite{ma2024feduc}\cite{nguyen2018sgd}.

\begin{definition}[(Gradient Divergence)]\label{def:gradientdivergence}
For $\forall v_i,\mathbf{w}$, we define $\xi_i$ as the upper bound of the gradient divergence between the dataset of worker $v_i$ and the globe dataset, \ie,
\begin{align}\notag
\|\nabla F(\mathbf{w})-\nabla F_i(\mathbf{w})\|\le\xi_i
\end{align}\notag
\end{definition}
\begin{definition}[(Local Gradient)]\label{def:localgradient}
$\mathbf{w}_i^*$ is the local optimal parameter vector on worker $v_i$, then we define the expectation of its random gradients on $v_i$ is
\begin{align}\notag
g_i^*=\mathbb{E}[\|\nabla F_i(\mathbf{w}_i^*;\xi_i)\|^2]
\end{align}
where $\xi_i$ is a arbitrary sample on worker $v_i$.
\end{definition}

For simplicity, we denote the minimum of the global loss function $F$ as $F^*$, short for $F(\mathbf{w}^*)$. 

\subsubsection{Convergence of Local Updating in One Round}\label{subsubsec:intraclusteranalysis}

We begin by presenting the following Lemma \ref{lem:basegap}, which details the convergence of the local update in one round. 

\begin{lemma}\label{lem:basegap}
Taking $\eta<\frac{\mu}{2L^2}$, by the local updating of Eq. \eqref{eq:updating}, it holds that
\begin{align}\notag
\mathbb{E}[F(\mathbf{w}_{t+1}^i)]-F^*
\le \rho(\mathbb{E}[F(\mathbf{\hat{w}}_t^i)]-F^*)+\delta_i\mbox{,}
\end{align}
where $\rho=1-\mu\eta$, and $\delta_i=\frac{\eta}{2}\xi_i+L\eta^2g_i^*$.
\end{lemma}
\begin{proof}
Based on Assumption \ref{ass:smoothness}, $F$ is a $L$-smooth function. By combining with Eq. \eqref{eq:updating}, we can deduce that
\begin{align}\notag
&F(\mathbf{w}_{t+1}^i)-F^* \\\notag
\le& F(\mathbf{\hat{w}}_t^i)-F^*+\langle\nabla F(\mathbf{\hat{w}}_t^i),\mathbf{w}_{t+1}^i-\mathbf{\hat{w}}_t^i\rangle+\frac{L}{2}\|\mathbf{w}_{t+1}^i-\mathbf{\hat{w}}_t^i\|^2 \\\notag
=& F(\mathbf{\hat{w}}_t^i)-F^*-\eta\langle\nabla F(\mathbf{\hat{w}}_t^i),\nabla F_i(\mathbf{\hat{w}}_t^i;\xi_t^i)\rangle \\\label{eq:fjwtjfw}
&+\frac{L\eta^2}{2}\|\nabla F_i(\mathbf{\hat{w}}_t^i;\xi_t^i)\|^2\mbox{.}
\end{align}
By Definition \ref{def:gradientdivergence}, we obtain the expectation of the third term of Eq. \eqref{eq:fjwtjfw} as
\begin{align}\notag
&\mathbb{E}[\langle\nabla F(\mathbf{\hat{w}}_t^i),\nabla F_i(\mathbf{\hat{w}}_t^i;\xi_t^i)\rangle] \\\notag
=&\langle\nabla F(\mathbf{\hat{w}}_t^i),\nabla F_i(\mathbf{\hat{w}}_t^i)\rangle \\\notag
=&\frac{1}{2}\biggl(\|\nabla F(\mathbf{\hat{w}}_t^i)\|^2+\|\nabla F_i(\mathbf{\hat{w}}_t^i)\|^2-\|\nabla F(\mathbf{\hat{w}}_t^i)-\nabla F_i(\mathbf{\hat{w}}_t^i)\|^2\biggr) \\\label{eq:Enablafproduct}
\ge&\frac{1}{2}\biggl(\|\nabla F(\mathbf{\hat{w}}_t^i)\|^2+\|\nabla F_i(\mathbf{\hat{w}}_t^i)\|^2-\xi_i^2\biggr)\mbox{,}
\end{align}
By Assumption \ref{ass:strongconvexity}, it is obvious that $F$ is $\mu$-strongly convex, which follows
\begin{align}\label{eq:censynconvexF}
\|\nabla F(\mathbf{\hat{w}}_t^i)\|^2\ge 2\mu(F(\mathbf{\hat{w}}_t^i)-F^*)\mbox{.}
\end{align}
Utilizing the AM-GM Inequality, we can express the expectation of the last term of Eq. \eqref{eq:fjwtjfw} as
\begin{align}\notag
&\mathbb{E}[\|\nabla F_i(\mathbf{\hat{w}}_t^i;\xi_t^i)\|^2] \\\notag
=&\mathbb{E}[\|\nabla F_i(\mathbf{\hat{w}}_t^i;\xi_t^i)-\nabla F_i(\mathbf{w}_i^*;\xi_t^i)+\nabla F_i(\mathbf{w}_i^*;\xi_t^i)\|^2] \\\label{eq:Enablaf}
\le&2\mathbb{E}[\|\nabla F_i(\mathbf{\hat{w}}_t^i;\xi_t^i)-\nabla F_i(\mathbf{w}_i^*;\xi_t^i)\|^2]+2\mathbb{E}[\|\nabla F_i(\mathbf{w}_i^*;\xi_t^i)\|^2]\mbox{,}
\end{align}
where $\mathbf{w}_i^*$ is the optimal solution of $F_i$. According to Lemma 3 of \cite{nguyen2018sgd}, we have
\begin{align}\label{eq:diffFL}
\mathbb{E}[\|\nabla F_i(\mathbf{\hat{w}}_t^i;\xi_t^i)-\nabla F_i(\mathbf{w}_i^*;\xi_t^i)\|^2]
\le2L(F_i(\mathbf{\hat{w}}_t^i)-F_i^*)\mbox{.}
\end{align}
Since $F_i$ is $\mu$-strongly convex for each $\forall v_i\in\mathcal{V}$, we have
\begin{align}\label{eq:convergenceFi}
F_i(\mathbf{\hat{w}}_t^i)-F_i^*\le\frac{1}{2\mu}\|\nabla F_i(\mathbf{\hat{w}}_t^i)\|^2 \mbox{.}
\end{align}
By taking Eqs. \eqref{eq:diffFL} and \eqref{eq:convergenceFi} into Eq. \eqref{eq:Enablaf}, we have
\begin{align}\label{eq:Enablaf2}
\mathbb{E}[\|\nabla F_i(\mathbf{\hat{w}}_t^i;\xi_t^i)\|^2]\le\frac{2L}{\mu}\|\nabla F_i(\mathbf{\hat{w}}_t^i)\|^2+2g_i^*\mbox{.}
\end{align}
By taking Eqs. \eqref{eq:Enablafproduct}, \eqref{eq:censynconvexF} and \eqref{eq:Enablaf2} into Eq. \eqref{eq:fjwtjfw}, we have
\begin{align} \notag
&\mathbb{E}[F(\mathbf{w}_{t+1}^i)]-F^* \\\notag
\le& (1-\mu\eta)(\mathbb{E}[F(\mathbf{\hat{w}}_t^i)]-F^*)-(\frac{\eta}{2}-\frac{L^2\eta^2}{\mu})\|\nabla F_i(\mathbf{\hat{w}}_t^i)\|^2 \\\label{eq:Enablaf3}
&+\frac{\eta}{2}\xi_i^2+L\eta^2g_i^*\mbox{.}
\end{align}
Since $\eta<\frac{\mu}{2L^2}$, we have
\begin{align}\label{eq:Enablaf4}
\mathbb{E}[F(\mathbf{w}_{t+1}^i)]-F^* \le \rho(\mathbb{E}[F(\mathbf{\hat{w}}_t^i)]-F^*)+\delta_i\mbox{,}
\end{align}
where $\rho=1-\mu\eta$ and $\delta_i=\frac{\eta}{2}\xi_i^2+L\eta^2g_i^*$.
\end{proof}

\subsubsection{A key Lemma for Analysis}\label{subsubsec:keylemma}

To simplify the expression, we define three sequences of
matrices $\mathbf{X}_t$, $\mathbf{Y}_t^j$ and $\mathbf{Z}_t$ for $t\ge 1$, where $\mathbf{X}_t=\mbox{diag}(x_t^1,x_t^2,...,x_t^N)$ and $\mathbf{Y}_t^j=\mbox{diag}(y_t^{1,j},y_t^{2,j},...,y_t^{N,j})$, $j\in[1,N]$ are $N\times N$ diagonal matrices. $\mathbf{Z}_t=[z_t^1,z_t^2,...,z_t^N]^\top$ is a $N$-dimensional vector.
In particular, if worker $v_i$ is not activated by the coordinator at round $t$, \ie, $v_i\notin\mathcal{A}_t$, then the $i$-th elements of $\mathbf{X}_t$, $\mathbf{Y}_t^j$ and $\mathbf{Z}_t$ are given by $x_t^i=1$, $y_t^{i,j}=0$ and $z_t^i=0$, respectively.
$x_t^i$ and $y_t^{i,j}$ satisfy $\theta_t^i=x_t^i+\sum_{v_j \in \mathcal{N}_t^i}y_t^{i,j}\le 1$.
We first state a key lemma for our statement.
For convenience, let $\omega_t^j=t-\tau_t^j-1\ge0 $, $\tau_{\max}=\max_{j,t}\{\tau_t^j\}$ and $\theta_{\max}=\max_{i,t}\{\theta_t^i\}$.

\begin{lemma}\label{lem:Qasy}
Let ${\mathbf{Q}_t}$ be a sequence of matrices for $t\ge0$. If $\mathbf{Q}_t\le \mathbf{X}_t\mathbf{Q}_{t-1}+\sum_{v_j \in \mathcal{N}_t^i}\mathbf{Y}_t^j\mathbf{Q}(\omega_t^j)+\mathbf{Z}_t$, then we have
\begin{equation}\notag
\mathbf{Q}_T\le \prod_{t=0}^T\mathbf{P}_t\mathbf{Q}_0+\sum_{t=0}^T\mathbf{\Delta}_t\mbox{,}
\end{equation}
where $\mathbf{P}_t=\mbox{diag}(p_t^1,p_t^2,...,p_t^N)$ is a $N\times N$ diagonal matrix and each scalar component satisfies
\begin{align}\notag
p_t^i=
\begin{cases}
{\theta_{\rm max}}^{\frac{1}{1+\tau_{max}}}, &\text{if } v_i\in\mathcal{A}_t \\
1, &\text{otherwise}
\end{cases}\mbox{.}
\end{align}
$\mathbf{\Delta}_t$ is a $M$-dimensional vector satisfies the following recursive relation:
\begin{align}\notag
\mathbf{\Delta}_t=
\begin{cases}
[0,0,...,0]^\top, t=0\\
(\mathbf{X}_t+\sum\nolimits_{v_j \in \mathcal{N}_t^i}\mathbf{Y}_t^j-\mathbf{E})\sum\limits_{r=0}^{t-1}\mathbf{\Delta}_r+\mathbf{Z}_t,t\ge 1\mbox{.}
\end{cases}
\end{align}
\end{lemma}
\begin{proof}
We first derive the following relation:
\begin{align}\label{eq:diagp}
\mathbf{X}_t+\sum_{v_j \in \mathcal{N}_t^i}\mathbf{Y}_t^j\prod_{r=t-\tau_t^j}^{t-1}[\mathbf{P}_r]^{-1}
\le \mbox{diag}(\hat{p}_t^1,\hat{p}_t^2,...,\hat{p}_t^N)
\end{align}
where
\begin{align}\notag
\hat{p}_t^i=
\begin{cases}
x_t^i+\sum_{v_j \in \mathcal{N}_t^i}y_t^{i,j}{\theta_{\rm max}}^{-\frac{\tau_{\rm max}}{1+\tau_{\rm max}}}, &\text{if } v_i\in\mathcal{A}_t \\
1, &\text{otherwise}
\end{cases}\mbox{.}
\end{align}
Since $\theta_{\rm max}<0$, ${\theta_{\rm max}}^{-\frac{\tau_{\rm max}}{1+\tau_{\rm max}}}>1$. It follows that
\begin{align}\notag
x_t^i+\sum_{v_j \in \mathcal{N}_t^i}y_t^{i,j}{\theta_{\rm max}}^{-\frac{\tau_{\rm max}}{1+\tau_{\rm max}}}
\le& (x_t^i+\sum_{v_j \in \mathcal{N}_t^i}y_t^{i,j}){\theta_{\rm max}}^{-\frac{\tau_{\rm max}}{1+\tau_{\rm max}}} \\\notag
\le& \theta_{\rm max}\cdot{\theta_{\rm max}}^{-\frac{\tau_{\rm max}}{1+\tau_{\rm max}}}\\\label{eq:decasylemmatemp}
=&{\theta_{\rm max}}^{\frac{1}{1+\tau_{\rm max}}}
\end{align}
By substituting Eq. \eqref{eq:decasylemmatemp} into Eq. \eqref{eq:diagp}, we derive that
\begin{align}
\mathbf{X}_t+\sum_{v_j \in \mathcal{N}_t^i}\mathbf{Y}_t^j\prod_{r=t-\tau_t^j}^{t-1}[\mathbf{P}_r]^{-1}\le\mathbf{P}_t
\end{align}
When $t=0$, it is obvious that Lemma \ref{lem:Qasy} is true. We assume that the induction hypothesis holds for all $t$ ranging from 0 to $T-1$, \ie,
\begin{equation}
\mathbf{Q}_t\le \prod_{r=0}^t\mathbf{P}_r\mathbf{Q}_0+\sum_{r=0}^t\mathbf{\Delta}_r,\quad\forall t\in\{0,1,...T-1\}\mbox{.}
\end{equation}
When $t=T$, we obtain that
\begin{align}\notag
\mathbf{Q}_T \le&\mathbf{X}_T\mathbf{Q}_{T-1}+\sum_{v_j \in \mathcal{N}_t^i}\mathbf{Y}_T^j\mathbf{Q}(\omega_T^j)+\mathbf{Z}_T \\\notag
\le& \mathbf{X}_T\biggl[\prod_{t=0}^{T-1}\mathbf{P}_t\mathbf{Q}_0+\sum_{t=0}^{T-1}\mathbf{\Delta}_t\biggr] \\\notag
&+\sum_{v_j \in \mathcal{N}_t^i}\mathbf{Y}_T^j\biggl[\prod_{t=0}^{\omega_T^j}\mathbf{P}_t\mathbf{Q}_0+\sum_{t=0}^{\omega_T^j}\mathbf{\Delta}_t\biggr]+\mathbf{Z}_T \\\notag
=& \biggl[\mathbf{X}_T+\sum_{v_j \in \mathcal{N}_t^i}\mathbf{Y}_T^j\prod_{t=T-\tau_T^j}^{T-1}[\mathbf{P}_t]^{-1}\biggr]\prod_{t=0}^{T-1}\mathbf{P}_t\mathbf{Q}_0 \\\notag
&+\biggl[\mathbf{X}_T\sum_{t=0}^{T-1}\mathbf{\Delta}_t+\sum_{v_j \in \mathcal{N}_t^i}\mathbf{Y}_T^j\sum_{t=0}^{\omega_T^j}\mathbf{\Delta}_t\biggr]+\mathbf{Z}_T \\\notag
\le& \mathbf{P}_T\prod_{t=0}^{T-1}\mathbf{P}_t\mathbf{Q}_0+\biggl[\mathbf{X}_T+\sum_{v_j \in \mathcal{N}_t^i}\mathbf{Y}_T^j\biggr]\sum_{t=0}^{T-1}\mathbf{\Delta}_t+\mathbf{Z}_T \\\label{eq:Enablaf6}
=& \prod_{t=0}^{T}\mathbf{P}_t\mathbf{Q}_0+\sum_{t=0}^T\mathbf{\Delta}_t\mbox{.}
\end{align}
Thus, Lemma \ref{lem:Qasy} holds based on the induction.
\end{proof}

\subsubsection{Convergence Bound of DySTop}\label{subsubsec:decasycon}

We next analyze DySTop's convergence bound with Lemma \ref{lem:basegap} and Lemma \ref{lem:Qasy} in the following Theorem \ref{thm:convergence}. For clarity, we represent the weight vector of all workers as $\mathbf{A}=[\alpha_1,...\alpha_N]$, where each $v_i$'s weight is $\alpha_i=\frac{D_i}{D}$, denoting the relative data size to the total.
Since a global model does not exist in a decentralized topology, we consider the weighted sum of of all local models, that is $\mathbf{w}_T=\sum_{v_i\in\mathcal{V}}\frac{D_i}{D}\mathbf{w}_T^i=\sum_{v_i\in\mathcal{V}}\alpha_i\mathbf{w}_T^i$.

\begin{theorem}\label{thm:convergence}
$\mathbf{w}_0$ is the initial model on each worker. After inter-cluster aggregation Eq. \eqref{eq:aggregation} is performed $T$ times, the weighted model $\mathbf{w}_T$ satisfies
\begin{align}\notag
\mathbb{E}[F(\mathbf{w}_T)]-F^*\le& \sum_{v_i\in\mathcal{V}}\alpha_i\rho^{\frac{\psi_iT}{1+\tau_{\max}}}(F(\mathbf{w}_0)-F^*) \\\notag
&+ \mathbf{A}\sum_{t=0}^T\mathbf{\Delta}_t\mbox{,}
\end{align}
where $\psi_i$ denotes the activating frequency of worker $v_i$, and $\mathbf{\Delta}_t$ meets the following recursive relation:
\begin{align}\label{eq:decasydelta}
&\mathbf{\Delta}_t=
\begin{cases}
[0,0,...,0]^\top, t=0\\
\mathbf{W}_t\sum\limits_{r=0}^{t-1}\mathbf{\Delta}_r+\mathbf{Z}_t, t\ge 1\mbox{.}
\end{cases}
\end{align}
where $\mathbf{W}_t=\mbox{diag}(w_t^1,w_t^2,...,w_t^N)$ is a $N\times N$ diagonal matrix and the scalar components satisfy
\begin{align}\notag
w_t^i=
\begin{cases}
\rho, &\text{if } v_i\in\mathcal{A}_t \\
1, &\text{otherwise}
\end{cases}\mbox{.}
\end{align}
and $\mathbf{Z}_t=[z_t^1,z_t^2,...,z_t^N]^\top$ is a vector and the scalar components satisfy
\begin{align}\notag
z_t^i=
\begin{cases}
\sum_{v_j \in \mathcal{N}_t^i}\sigma_t^{i,j}\delta_j, &\text{if } v_i\in\mathcal{A}_t \\
0, &\text{otherwise}
\end{cases}\mbox{.}
\end{align}
\end{theorem}
\begin{proof}
Due to the convexity of $F$ and $\sigma_t^{i,j}\in(0,1)$, together with Eqs. \eqref{eq:asyw} and \eqref{eq:aggregation}, for $\forall v_i\in\mathcal{A}_t$, it holds that
\begin{align}\notag
F(\mathbf{\hat{w}}_t^i)-F^*=&\sum_{v_j \in \mathcal{N}_t^i}\sigma_t^{i,j}F(\mathbf{\hat{w}}_t^i)-F^* \\\notag1
\le& \sum_{v_j \in \mathcal{N}_t^i}\sigma_t^{i,j}(F(\mathbf{w}_t^j)-F^*) \\\label{eq:Enablaf7}
=&\sum_{v_j \in \mathcal{N}_t^i}\sigma_t^{i,j}(F(\mathbf{w}_{t-\tau_t^j}^j)-F^*)\mbox{.}
\end{align}
According to Lemma \ref{lem:basegap}, we derive that
\begin{align}\notag
&\mathbb{E}[F(\mathbf{\hat{w}}_t^i)]-F^*\\\label{eq:Enablaf8}
\le& \rho\sum_{v_j \in \mathcal{N}_t^i}\sigma_t^{i,j}(\mathbb{E}[F(\mathbf{\hat{w}}_{t-\tau_t^j-1}^j)]-F^*)+\sum_{v_j \in \mathcal{N}_t^i}\sigma_t^{i,j}\delta_j\mbox{.}
\end{align}
Let $Q_t^i = F(\mathbf{\hat{w}}_t^i)-F^*$ and $\mathbf{Q}_t = [Q_t^1, ..., Q_t^N]^\top$. The recursive relation Eq. \eqref{eq:Enablaf8} at round $t$ is converted into
\begin{align}\label{eq:Enablaf9}
\mathbf{Q}_t\le \mathbf{X}_t\mathbf{Q}_{t-1}+\sum\nolimits_{v_j \in \mathcal{N}_t^i}\mathbf{Y}_t^j\mathbf{Q}(\omega_t^j)+\mathbf{Z}_t\mbox{,}
\end{align}
where $\mathbf{X}_t=\mbox{diag}(x_t^1,x_t^2,...,x_t^N)$ and $\mathbf{Y}_t^j=\mbox{diag}(y_t^{1,j},y_t^{2,j},...,y_t^{N,j})$ are diagonal matrices, and $\mathbf{Z}_t=[z_t^1,z_t^2,...,z_t^N]^\top$ is a vector. The scalar components are given by
\begin{align}\notag
x_t^i=
\begin{cases}
\rho\sigma_t^{i,j}, &\text{if } v_i\in\mathcal{A}_t \\
1, &\text{otherwise}
\end{cases}\mbox{,}
y_t^{i,j}=
\begin{cases}
\rho\sigma_t^{i,j}, &\text{if } v_i\in\mathcal{A}_t \\
0, &\text{otherwise}
\end{cases}
\end{align}
and
\begin{align}\notag
z_t^i=
\begin{cases}
\sum_{v_j \in \mathcal{N}_t^i}\sigma_t^{i,j}\delta_j, &\text{if } v_i\in\mathcal{A}_t \\
0, &\text{otherwise}
\end{cases}\mbox{.}
\end{align}
By leveraging Lemma \ref{lem:Qasy}, we derive that
\begin{align}
\mathbf{Q}_T\le \prod_{t=0}^T\mathbf{P}_t\mathbf{Q}_0+\sum_{t=0}^T\mathbf{\Delta}_t\mbox{,}
\end{align}
where $\mathbf{P}_t=\mbox{diag}(p_t^1,p_t^2,...,p_t^N)$ is a $N\times N$ diagonal matrix and the scalar components satisfy
\begin{align}\notag
p_t^i=
\begin{cases}
{\theta_{\rm max}}^{\frac{1}{1+\tau_{\rm max}}}, &\text{if } v_i\in\mathcal{A}_t \\
1, &\text{otherwise}
\end{cases}\mbox{,}
\end{align}
and $\mathbf{\Delta}_t$ satisfies \eqref{eq:decasydelta}.
When the training round $T$ is large enough, the models on each device converge uniformly, thus we have $\mathbf{w}_T^i\approx\mathbf{\hat{w}}_T^i$. Therefore,
\begin{align}\notag
&\mathbb{E}[F(\mathbf{w}_T)]-F^*\\\notag
\le& \sum_{v_i\in\mathcal{V}}\alpha_i(\mathbb{E}[F(\mathbf{w}_T^i)]-F^*)\\\notag
\approx& \sum_{v_i\in\mathcal{V}}\alpha_i(\mathbb{E}[F(\mathbf{\hat{w}}_T^i)]-F^*)\\\notag
=&\sum_{v_i\in\mathcal{V}}\alpha_iQ_T^i\\\notag
=& \mathbf{A}\mathbf{Q}_T \\\notag
\le& \mathbf{A}\prod_{t=0}^T\mathbf{P}_t\mathbf{1}_N(F(\mathbf{w}_0)-F^*)+ \mathbf{A}\sum_{t=0}^T\mathbf{\Delta}_t\\\label{eq:Enablaf11}
=&\sum_{v_i\in\mathcal{V}}\alpha_i\rho^{\frac{\psi_iT}{1+\tau_{\max}}}(F(\mathbf{w}_0)-F^*) + \mathbf{A}\sum_{t=0}^T\mathbf{\Delta}_t \mbox{.}
\end{align}
where $\psi_i$ denotes the activating frequency of worker $v_i$, so $\psi_i T$ is the number of times $v_i$ is activated over $T$ rounds.
Therefore, we complete the proof.
\end{proof}

\subsection{Discussions}

We denote the convergence bound of DySTop in Theorem \ref{thm:convergence} as $Bound_T=\sum_{v_i\in\mathcal{V}}\alpha_i\rho^{\frac{\psi_iT}{1+\tau_{\max}}}(F(\mathbf{w}_0)-F^*) + \mathbf{A}\sum_{t=0}^T\mathbf{\Delta}_t$. From this, we can derive several meaningful corollaries.

\begin{corollary}\label{cor:tau}
The convergence bound $Bound_T$ decreases as the upper bound of staleness $\tau_{\max}$ decreases.
\end{corollary}
\begin{corollary}\label{cor:psi}
The convergence bound $Bound_T$ decreases as the activating frequency $\psi_i$ of each worker $v_i$ increases.
\end{corollary}
\begin{corollary}\label{cor:noniid}
According to Eqs. \eqref{eq:Enablaf4} and \eqref{eq:decasydelta}, $\mathbf{\Delta}_t$ is partially influenced by the data distribution. As the degree of data non-IID among workers increases, the value of $\xi_i$ for each worker $v_i$ rises, leading to a high convergence bound $Bound_T$.
\end{corollary}

Corollary \ref{cor:tau} indicates that convergence performance can be improved by controlling the maximum staleness during training, ensuring it does not exceed a certain upper bound, \ie, $\tau_{\max}\le\tau_{\rm bound}$. Corollary \ref{cor:psi} suggests that increasing the number of activated workers $|\mathcal{A}_t|$ in each round can also improve convergence performance. However, this does not necessarily lead to a short convergence time, since each round's completion time depends on the model aggregations and training of all activated workers. Therefore, activating more workers may prolong the duration of each round. Corollaries \ref{cor:tau} and \ref{cor:psi} highlight the importance of carefully determining the worker set $|\mathcal{A}_t|$ to control staleness and accelerate convergence.
From Corollary \ref{cor:noniid}, when the data distribution among workers is IID, \ie, $\xi_i=0$ for $\forall v_i\in \mathcal{V}$, the convergence performance of DySTop can be further improved. In practice, however, data is often non-IID for corollary workers. To address this, one can design the training topology and select neighbors for each activated worker such that the union of a worker and its neighbors' datasets approximates an IID distribution, thereby enhancing training performance \cite{bellet2022d}.

Corollaries \ref{cor:tau}-\ref{cor:noniid} theoretically reveal the significance of both staleness control and topology construction in achieving better training performance, which will be elaborated in Section \ref{sec:algorithm}.

\section{Algorithm Description}\label{sec:algorithm}
\subsection{Algorithm Overview}

Tackling the problem $\textbf{P1}$ is a complex endeavor due to three main challenges below:
\begin{itemize}
    \item First, $\textbf{P1}$ is a round-coupled problem. The staleness constraints of workers are coupled with worker activation strategies across rounds.
    \item Second, based on Corollaries \ref{cor:tau} and \ref{cor:psi}, there is a trade-off between staleness control and convergence rate with different activated workers, but it's hard to quantify the correlation between them.
    \item Third, based on Corollary \ref{cor:noniid}, a lower degree of data non-IID leads to faster convergence. However, communication constraint forces us to select only a subset of neighbors, and improper selection can conversely exacerbate the degree of data non-IID.
\end{itemize}

To overcome these challenges, we first employ Lyapunov optimization to decouple the challenging round-coupled optimization problem into deterministic per-round subproblems. This prevents cross-round decision coupling from compromising global optimality.
We further solve each per-round subproblem in two steps. First, we devise a worker activation algorithm (WAA) to determine the active set $\mathcal{A}_t$ in each round, where we utilize Lyapunov drift-plus-penalty function to quantify the trade-off between staleness control and convergence rate. After that, we design a phase-aware topology construction algorithm (PTCA) to dynamically adapt the topology $\mathcal{G}_t$ in each round. We set individual strategies for different training phases to reduce communication overhead and non-IID degree.

\begin{figure}
    \centering
    \includegraphics[width=0.98\linewidth]{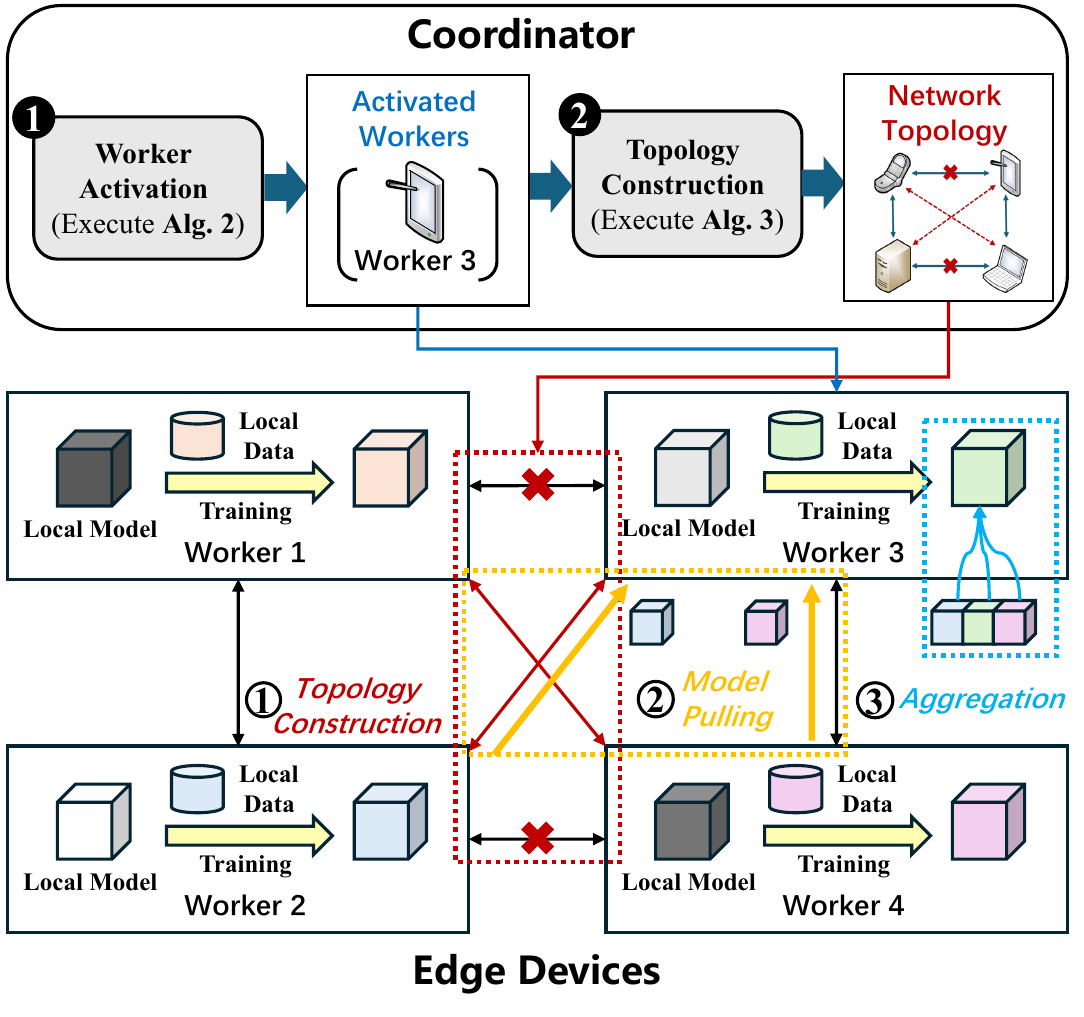}
    \caption{Workflow Illustration of DySTop}
    \label{fig:fram}
\end{figure}

Fig. \ref{fig:fram} shows a four-worker example to demonstrate the workflow of DySTop. There is a network topology with 4 workers ($\mathcal{V}=\{v_1, v_2, v_3, v_4\}$). At a certain round $t$, the coordinator first collects the information of workers (\textit{e.g.}, current model staleness and estimated training time of the workers) and performs WAA to determine the active set $\mathcal{A}_t=\{v_3\}$. Next, the coordinator performs PTCA to construct the topology and notifies workers to adapt their neighbors. In our example, the directed communications of $v_1$-$v_3$ and $v_2$-$v_4$ are removed. Instead, directed communications $v_2$-$v_3$ and $v_1$-$v_4$ are added. Then $v_3$ pulls models from neighbors $v_1$ and $v_4$ to perform aggregation with its own local model.
This workflow is executed in all rounds to solve the decoupled subproblems.

\subsection{Lyapunov Optimization Based Problem Transformation}\label{subsec:Lyapunov}
The core idea of Lyapunov optimization is transforming the staleness constraint into a queue stability problem. Following this idea,
For each worker $v_i$, we create a staleness queue $q_t^i$ to track the difference between its cumulative staleness up to round $t$ and the given upper bound of staleness $\tau_{\rm bound}$. $q_t^i$ is updated by the following recurrence relation:
\begin{align}\label{eq:queue}
q_{t+1}^i=\max\{q_t^i+\tau_t^i-\tau_{\rm bound},0\}\mbox{.}
\end{align}
where $q_0^i$ is initialized as 0.
By employing Lyapunov optimization theory, we can obtain Theorem \ref{thm:dpp} about $\textbf{P1}$ as follows,
\begin{theorem}\label{thm:dpp}
Approximating \textnormal{\textbf{P1}} involves solving subproblem \textnormal{\textbf{P2}} at each round $t,\forall t\in\mathcal{R}$
\begin{align}\label{eq:lya}
\textnormal{\textbf{(P2)}}:&\min_{\boldsymbol{a}_t, \boldsymbol{c}_t}\sum_{v_i\in\mathcal{V}}q_t^i(\tau_t^i-\tau_{\rm bound})+V\cdot H_t  \\\notag
{\st} \quad &\eqref{cons:loss}-\eqref{cons:topo}
\mbox{.}
\end{align}
where Eq. \eqref{eq:lya} is a drift-plus-penalty function, and $V$ serves as a trade-off factor, balancing the staleness queue stabilization against the minimization of training duration.
The expected training duration obtained by solving \textnormal{\textbf{P2}} has a minimum gap
\begin{equation}
    \sum_{t=1}^T \mathbb{E}\left[H_t\vert\Theta_t\right] \leq \frac{T\cdot\Gamma}{V} + H^*,
\end{equation}
where $H^*$ represents the optimal total training duration under the optimal strategies, and $\Gamma$ is a positive constant.
\end{theorem}
\begin{proof}
The Lyapunov function is defined as
\begin{align}\label{eq:lf}
    L(\Theta_t)=\frac{1}{2}\sum_{v_i\in\mathcal{V}}(q_t^i)^2\mbox{.}
\end{align}
where $\Theta_t=\{q_t^i\}_{v_i\in\mathcal{V}}$ records the staleness queues of all workers.
For stability analysis, we formulate the single-round Lyapunov drift function as
\begin{align}\label{eq:drift1round}
    \Delta \Theta_t\triangleq \mathbb{E}[L(\Theta_{t+1})-L(\Theta_t)|\Theta_t]\mbox{.}
\end{align}
Based on Eq. \eqref{eq:queue}, we derive that
\begin{align}\notag
    \left[q_{t+1}^i\right]^2\leq& \left[q_t^i + (\tau_t^i - \tau_{\rm bound})\right]^2 \\\notag
    =& \left[q_t^i\right]^2 + \left(\tau_t^i - \tau_{\rm bound}\right)^2 + 2q_i(t)\left(\tau_t^i - \tau_{\rm bound}\right) \\
    \leq& \left[q_t^i\right]^2 + \left[\tau_t^i\right]^2 + (\tau_{\rm bound})^2 + 2q_i(t)\left(\tau_t^i - \tau_{\rm bound}\right).
\end{align}
We define that
\begin{align}\notag
&\Delta \Theta_t\triangleq \mathbb{E}[L(\Theta_{t+1})-L(\Theta_t)|\Theta_t] \\\notag
=&\mathbb{E}\left[\sum_{v_i\in\mathcal{V}}\frac{(q_{t+1}^i)^2-(q_t^i)^2}{2}\Bigg\vert \Theta_t\right]\\\notag
\leq&\mathbb{E}\left[\sum_{v_i\in\mathcal{V}}\frac{(\tau_t^i)^2+\tau_{\rm bound}^2}{2}\Bigg\vert \Theta_t\right]+\mathbb{E}\left[\sum_{v_i\in\mathcal{V}}q_t^i(\tau_t^i-\tau_{\rm bound})\Bigg\vert \Theta_t\right]\\\label{eq3}
=&\Gamma+\mathbb{E}\left[\sum_{v_i\in\mathcal{V}}q_t^i(\tau_t^i-\tau_{\rm bound})\Bigg\vert \Theta_t\right]\mbox{.}
\end{align}
where $\Gamma$ is a positive constant with upper bound $N\cdot\tau_{\rm bound}^2$, which is calculated as follows:
\begin{align}\notag
    \mathbb{E}\left[\sum_{v_i\in\mathcal{V}}\frac{(\tau_t^i)^2+\tau_{\rm bound}^2}{2}\Bigg\vert \Theta(t)\right]
    \leq& \sum_{v_i\in\mathcal{V}}\frac{\tau_{\rm bound}^2+\tau_{\rm bound}^2}{2}\\
    =&N\cdot\tau_{\rm bound}^2\mbox{.}
\end{align}
Let's accumulate the sum of drift functions over $T$ rounds
\begin{equation}\label{eq:drift}
    L\left(\Theta_T\right)-L\left(\Theta_1\right)\leq T\cdot\Gamma+\sum_{t=1}^T\sum_{v_i\in\mathcal{V}}q_t^i(\tau_t^i-\tau_{\rm bound})\mbox{.}
\end{equation}
Due to the constraint \eqref{cons:stale}, we take the expectation operation on both sides of
\eqref{eq:drift} and obtain the following expression
\begin{align}
    \mathbb{E}\left[L(\Theta_T)-L(\Theta_1)\right]\leq T\cdot\Gamma-\beta\sum_{t=1}^T\sum_{v_i\in\mathcal{V}}\mathbb{E}\left[q_t^i\right]\mbox{.}
\end{align}
where $\beta$ is a positive constant in the neighborhood of 0. Since $L(\Theta_1)=0$ and $L(\Theta_T)\geq0$, we have
\begin{equation}
    \frac{1}{T}\sum_{t=1}^{T}\sum_{v_i\in\mathcal{V}}\mathbb{E}\left[q_t^i\right]\leq \frac{\Gamma}{\beta}\mbox{.}
\end{equation}
This inequality indicates that under the constraint \eqref{cons:stale}, queue stability can be guaranteed. Next, we assume that there is an optimal total training duration for the $\textbf{P1}$ $H^*=\sum_{t=1}^T H^{*}_t$. We introduce the $\textit{drift-plus penalty function}$ as
\begin{align}\notag
&\Delta \Theta_t+V\cdot \mathbb{E}[H_t|\Theta_t] \\\notag
\leq& \Gamma + \mathbb{E}\left[\sum_{v_i\in\mathcal{V}}q_t^i(\tau_t^i-\tau_{\rm bound})\Bigg\vert\Theta_t\right]+V\cdot\mathbb{E}\left[H_t\vert\Theta_t\right] \\\notag
\leq& \Gamma + \mathbb{E}\left[\sum_{v_i\in\mathcal{V}}q_t^i(\tau_t^i-\tau_{\rm bound})\Bigg\vert\Theta_t\right]+V\cdot H^{*}_t \mbox{.}
\end{align}
where $V$ is a positive hyperparameter to balance the training duration minimization.
The above inequality is summed on both sides for $T$ rounds, yielding the following results, respectively:
\begin{align}\notag
    &L(\Theta_T)-L(\Theta_1) + V\sum_{t=1}^T \mathbb{E}\left[H_t\vert\Theta_t\right] \\\notag
    \leq& T\cdot\Gamma+\sum_{t=1}^T\sum_{v_i\in\mathcal{V}}q_t^i(\tau_t^i-\tau_{\rm bound}) + V\sum_{t=1}^T H^*_t\mbox{.}
\end{align}
Due to the constraint \eqref{cons:stale}, it holds that $\sum_{t=1}^T\sum_{v_i\in\mathcal{V}}q_t^i(\tau_t^i-\tau_{\rm bound})<0$, together with $L(\Theta(T))\geq 0$ and $L(\Theta(1))=0$, we deduce that
\begin{equation}
    V\sum_{t=1}^T \mathbb{E}\left[H_t\vert\Theta_t\right] \leq T\cdot\Gamma + V\cdot H^*
\end{equation}
Divide both sides by $V$, then we can obtain the expected gap. Thus, Theorem \ref{thm:dpp} holds.
\end{proof}

The optimization objective of $\textbf{P2}$ is to independently determine the worker activation strategy $\boldsymbol{a}_t$ and the topology construction strategy $\boldsymbol{c}_t$ for each round, aiming to minimize the weighted sum of the workers' current staleness deviation and round duration. 

\subsection{Worker Activation Algorithm (WAA)}\label{subsec:workerselection}

To address $\textbf{P2}$, we first propose the WAA to determine the worker activation strategy $\boldsymbol{a}_t$ at each round, as described in Alg. \ref{alg:device}.
The drift-plus-penalty function Eq. \eqref{eq:lya} provides the quantification of the trade-off between staleness control and single-round duration.
In addition, Eq. \eqref{eq:duration} inspires us to prioritize workers with smaller $H_t^i$, which can help control staleness while reducing single-round duration $H_t$.
Consequently, the target of WAA is to minimize Eq. \eqref{eq:lya} by activating an appropriate number of workers with smaller $H_t^i$.
At the beginning of each round $t$, we initialize the active set $\mathcal{A}_t$ as empty, and set all $a_t^i\in \boldsymbol{a}_t$ to 0 (Line \ref{alg:line:init}). Then, we sort workers in $\mathcal{V}$ in the ascending order of the sum of training time and transmission time $H_t^i$, and the set of sorted workers is denoted as $\mathcal{V}_{\rm sort}$ (Line \ref{alg:line:sort}).
Next, we sequentially add workers from $\mathcal{V}_{\rm sort}$ to $\mathcal{A}_t$. For each progressively larger active set $\mathcal{A}_t$, up to the maximum worker count, we pre-update each worker's staleness and compute the corresponding value of Eq. \eqref{eq:lya}, denoted as $S_{\mathcal{A}_t}$ (Lines \ref{alg:line:addstart}-\ref{alg:line:est}). The minimum value $S_{\min}$ and the corresponding set $\mathcal{A}^*_t$ are recorded (Lines \ref{alg:line:recordstart}-\ref{alg:line:recordend}). Finally, we set each $a_t^i$ whose worker is in $\mathcal{A}^*_t$ to 1 (Lines \ref{alg:line:abegin}-\ref{alg:line:aend}) and return the worker activation strategy $\boldsymbol{a}_t$ (Line \ref{alg:line:return}).

\begin{algorithm}[t]
\renewcommand{\algorithmicrequire}{\textbf{Input:}}
\renewcommand{\algorithmicensure}{\textbf{Output:}}
\caption{Worker Activation Algorithm (WAA)}\label{alg:device}
\begin{algorithmic}[1]
\REQUIRE {Staleness degree $\tau_t^i, \forall v_i\in \mathcal{V}$, value of virtual queue $q_t^i, \forall v_i\in \mathcal{V}$, sum of training time and transmission time $H_t^i, \forall v_i\in \mathcal{V}$}
\ENSURE {The worker activation strategy $\boldsymbol{a}_t=\{a_t^i|v_i\in\mathcal{V}\}$ at round $t$}
\STATE {\textbf{Initialize:} $\mathcal{A}_t=\varnothing$, $a_t^i=0,\forall a_t^i\in\boldsymbol{a}_t$, $S_{\min}=+\infty$;}\label{alg:line:init}
\STATE {$\mathcal{V}_{\rm sort}\leftarrow$ sort $\mathcal{V}$ in the ascending order of $H_t^i$;}\label{alg:line:sort}
\FOR{$K=1$ to $N$}\label{alg:line:addstart}
\STATE {$\mathcal{A}_t=\mathcal{A}_t \cup \mathcal{V}_{\rm sort}\left[K\right]$;}\label{alg:line:addend}
\STATE {$S_{\mathcal{A}_t}\leftarrow$ calculate Eq. \eqref{eq:lya} for given $\mathcal{A}_t$;}\label{alg:line:est} 
\IF{$S_{\mathcal{A}_t} < S_{\rm min}$}\label{alg:line:recordstart}
\STATE {$S_{\rm min} = S_{\mathcal{A}_t}$;}
\STATE {$\mathcal{A}^*_t=\mathcal{A}_t$;}
\ENDIF\label{alg:line:recordend}
\ENDFOR
\FOR{\textbf{each} $v_i\in\mathcal{A}^*_t$}\label{alg:line:abegin}
\STATE {$a_t^i=1$;}
\ENDFOR\label{alg:line:aend}
\STATE {\textbf{Return:} $\boldsymbol{a}_t=\{a_t^i|v_i\in\mathcal{V}\}$.}\label{alg:line:return}
\end{algorithmic}
\end{algorithm}

\subsection{Phase-aware Topology Construction Algorithm (PTCA)}\label{subsec:dynamictopo}
We develop the PTCA to construct the network topology by selecting neighbors for each activated worker $v_i\in \mathcal{A}_t$.
We partition the overall training process into two phases. Then, for each phase, we independently design the strategy for selecting neighbors based on different objectives.

\subsubsection{Phase 1: Deal with Non-IID Data}
In early training phases, non-IID data significantly slows down model convergence \cite{wang2025fedsiam}.
Motivated by \cite{bellet2022d}, we pair workers with significantly different data distributions as neighbors. This can help their combined datasets better represent all classes, in turn mitigating model divergence (an example is shown in Fig. \ref{fig:Non-IID}).

Earth Mover's Distance (EMD) \cite{rubner2000earth} is a general measure  to quantify the difference of two datasets' distributions.
We denote the EMD between $v_i$ and $v_j$'s datasets $\mathcal{D}_i$ and $\mathcal{D}_j$ as
\begin{equation}\label{eq:EMD}
    \text{EMD}(\mathcal{D}_i, \mathcal{D}_j)=\sum_{c_k\in\mathcal{C}} \|\frac{D_i^k}{D_i} - \frac{D_j^k}{D_j}\|.
\end{equation}
where $c_k$ is the $k$-th class of the global dataset $\mathcal{D}$ and the set $\mathcal{C}$ consists of classes of data in $\mathcal{D}$. $D_i^k$ represents the size of the $k$-th class data in $v_i$'s dataset $\mathcal{D}_i$. In addition, given that worker-to-worker distance influences model transmission latency, we also take physical distance into account. Let $\text{Dist}(v_i, v_j)$ denote the physical distance between workers $v_i$ and $v_j$. We define the priority for $v_i$ to select $v_j$ as a neighbor as follows
\begin{equation}\label{eq:priority1}
    p_1(v_i, v_j) = \frac{\text{EMD}(\mathcal{D}_i, \mathcal{D}_j)}{\text{EMD}_{\max}} + \left(1-\frac{\text{Dist}(v_i, v_j)}{\text{Dist}_{\max}}\right)\mbox{.}
\end{equation}
where $\text{EMD}_{\max}$ and $\text{Dist}_{\max}$ indicate the maximum EMD and distance among workers.
\begin{figure}[h]
    \centering
    \includegraphics[width=0.7\linewidth]{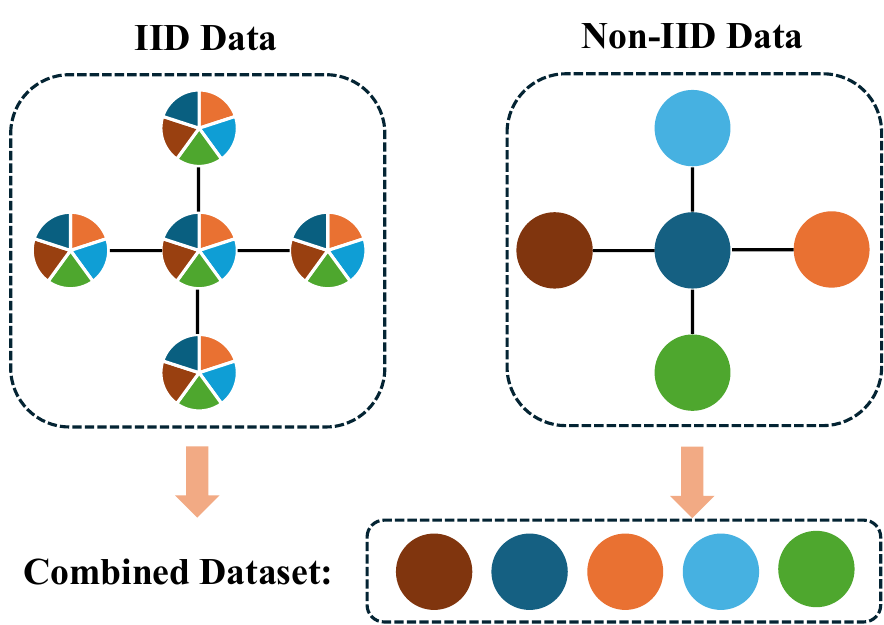}
    \caption{A five-worker example of IID and non-IID scenarios, where each color represents a different class of the dataset. In the IID scenario (left), all classes are in equal proportions on all workers, while in the non-IID scenario (right), each worker has only one class. However, both scenarios' combined datasets of all workers represent the same distribution.}
    \label{fig:Non-IID}
    %\vspace{-5mm}
\end{figure}

\subsubsection{Phase 2: Diverse Neighbors and Staleness Control}

As the training progresses, the impact of staleness increases \cite{sun2024staleness,liao2023adaptive}. At the same time, diverse neighbor selections play a crucial role in enhancing accuracy \cite{de2023epidemic}.
Therefore, we design the following priority for $v_i$ to select $v_j$ as a neighbor
\begin{equation}\label{eq:priority2}
    p_2(v_i, v_j)=\left(1-\frac{\text{Pull}(v_i, v_j)}{t}\right)\frac{1}{1+\vert \tau_i-\tau_j\vert}\mbox{.}
\end{equation}
where $\text{Pull}(v_i, v_j)$ records the times of $v_i$ pulling models from $v_j$. The term $\left(1-\frac{\text{Pull}(v_i, v_j)}{t}\right)$ encourages the selection of workers with fewer pull requests, and the term $\frac{1}{1+\vert \tau_i-\tau_j\vert}$ ensures that the staleness gap between workers is not too large.

\begin{figure}[t]
    \centering
\subfigure{\label{fig:dnns_lat_wifi}\includegraphics[width=.49\linewidth]{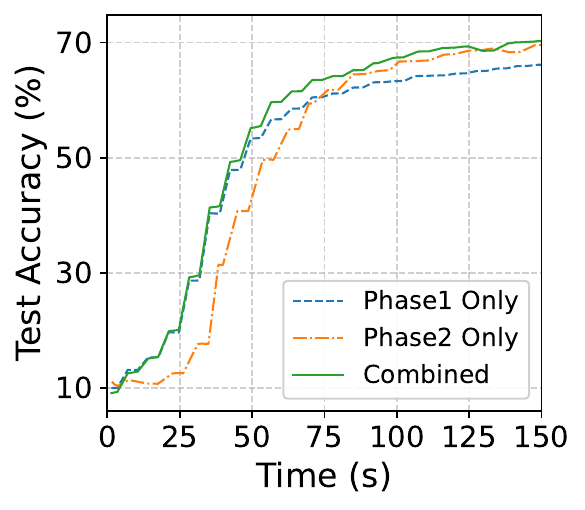}}
    \subfigure{\label{fig:dnns_lat_5g}\includegraphics[width=.49\linewidth]{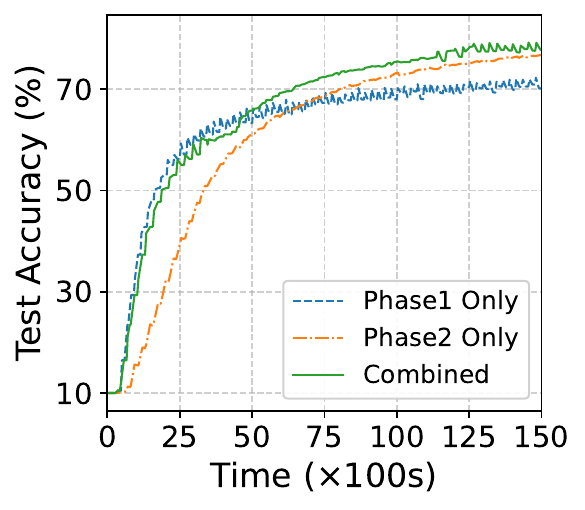}}\vspace{-0.1in}
    \small{\caption{Comparison of different priority-based topology construction strategies on non-IID datasets. \textit{Left}: FMNIST; \textit{Right}: CIFAR-10.}\label{fig:observation}}
    \vspace{-0.5cm}
\end{figure}

\subsubsection{Algorithm design for PTCA}

The goal of our proposed PTCA is to determine the topology construction strategy $\boldsymbol{c}_t$ for each round $t$, which is formally described in Alg. \ref{alg:topology}. Let $C_t^i$ denote the workers within $v_i$'s communication range include itself at round $t$.
To maximize bandwidth utilization subject to each worker's bandwidth budget, Alg. \ref{alg:topology} iteratively selects the highest-priority neighbor for each activated worker $v_i\in\mathcal{A}_t$. This process continues until the bandwidth of workers in $C_t^i$ is fully consumed.
Specifically, we first initialize each worker's in-neighbor and out-neighbor sets as empty sets, and set all $c_t^{i,j}\in\boldsymbol{c}_t$ to 0. $B_{\rm tmp}$ is maintained to track the cumulative bandwidth consumption of all workers (Line \ref{alg:ptca:init}). Then, we sort $C_t^i$ in descending order of priority for each activated worker, where the corresponding phase determines the priority (Lines \ref{alg:ptca:sortbegin}-\ref{alg:ptca:sortend}). Next, for each iteration (Lines \ref{alg:ptca:iterbegin}-\ref{alg:ptca:iterend}), each worker $v_i\in \mathcal{A}_t$ first confirms sufficient bandwidth resources for model pulling (Lines \ref{alg:ptca:checkbegin}-\ref{alg:ptca:checkend}), then selects the top-ranked worker from sorted $C_t^i$ with available bandwidth for model transmission, and updates neighbor sets and bandwidth consumption (Lines \ref{alg:ptca:selectbegin}-\ref{alg:ptca:selectend}). PTCA terminates iteration and returns $\boldsymbol{c}_t$ when the total bandwidth consumption remains unchanged between two consecutive iterations (Lines \ref{alg:ptca:termbegin}-\ref{alg:ptca:return}).

To verify the validation of our phase-aware strategy, we compare our PTCA with three settings: 1) Phase 1-Only (PTCA that only uses $p_1(v_i,v_j)$), 2) Phase 2-Only (PTCA that only uses $p_2(v_i,v_j)$), 3) Combined (\ie, Alg. \ref{alg:topology}).
We experiment with 100 workers to train a CNN model on FMNIST and ResNet-18 on CIFAR-10 with non-IID data.
The results of the experiment in Fig. \ref{fig:observation} show that
Phase 1-Only trains faster initially but converges to lower accuracy, while Phase 2-Only trains slower initially but achieves higher accuracy. The Combined setting achieves both faster convergence and higher ultimate accuracy compared to both the above two settings.

\begin{algorithm}[t]
\renewcommand{\algorithmicrequire}{\textbf{Input:}}
\renewcommand{\algorithmicensure}{\textbf{Output:}}
\caption{Phase-aware Topology Construction Algorithm (PTCA)}\label{alg:topology}
\begin{algorithmic}[1]
\REQUIRE {Active set $\mathcal{A}_t$ and lists of workers in $v_i$'s communication range $C_t^i, v_i\in \mathcal{V}$}
\ENSURE {The topology construction strategy $\boldsymbol{c}_t=\{c_t^{i,j}|v_i,v_j\in\mathcal{V}\}$ at round $t$}
\STATE {\textbf{Initialize:} $\mathcal{N}_t^i=\hat{\mathcal{N}}_t^i=\varnothing$, $B_t^i=0$, $\forall v_i \in \mathcal{V}$, $B_{\rm tmp}=0$, $c_t^{i,j}=0, \forall c_t^{i,j}\in\boldsymbol{c}_t$;}\label{alg:ptca:init}
\IF{$t \leq t_{\rm thre}$}\label{alg:ptca:sortbegin}
\STATE {Sort $C_t^i$ in the descending order of $p_1(v_i,v_j)$, $\forall v_i\in \mathcal{A}_t$;}
\ELSE
\STATE {Sort $C_t^i$ in the descending order of $p_2(v_i,v_j)$, $\forall v_i\in \mathcal{A}_t$;}
\ENDIF\label{alg:ptca:sortend}
\WHILE{true}\label{alg:ptca:iterbegin}
\FOR{\textbf{each} $v_i\in\mathcal{A}_t$}
\IF{$B_t^i+b>\hat{B}_t^i$}\label{alg:ptca:checkbegin}
\STATE {\textbf{continue};}
\ENDIF\label{alg:ptca:checkend}
\WHILE{$len(C_t^i)>0$}\label{alg:ptca:selectbegin}
\IF{$B_t^{C_t^i[0]}+b > \hat{B}_t^{C_t^i[0]}$}
\STATE {$C_t^i=C_t^i-C_t^i[0]$;}
\ELSE
\STATE {$c_t^{{i, C_t^i[0]}}=1$, $\mathcal{N}_t^i\cup\{C_t^i[0]\}$, $\hat{\mathcal{N}}_t^{C_t^i[0]}\cup\{v_i\}$;}
\STATE {$B_t^i=B_t^i+b$, $B_t^{C_t^i[0]}= B_t^{C_t^i[0]}+b$;}
\STATE {$C_t^i=C_t^i-C_t^i[0]$;}
\STATE {\textbf{break};}
\ENDIF
\ENDWHILE\label{alg:ptca:selectend}
\ENDFOR
\IF{$B_{\rm tmp}-\sum_{v_i\in\mathcal{V}}B_t^{i}=0$}\label{alg:ptca:termbegin}
\STATE {\textbf{break};}
\ELSE
\STATE {$B_{\rm tmp}=\sum_{v_i\in\mathcal{V}}B_t^{i}$}
\ENDIF\label{alg:ptca:termend}
\ENDWHILE\label{alg:ptca:iterend}
\STATE {\textbf{Return:} $\boldsymbol{c}_t=\{c_t^i|v_i,v_j\in\mathcal{V}\}$.}\label{alg:ptca:return}

\end{algorithmic}
\end{algorithm}

\section{Simulation Experiment}\label{sec:evaluation}

In this section, we simulate a large-scale DFL in mobile edge network to verify the effectiveness of our proposed mechanism and algorithms.

\subsection{System Setup}\label{subsec:setup}

\begin{figure*}[t]
\begin{minipage}[t]{0.49\linewidth}
\includegraphics[width=0.49\linewidth,height=3.6cm]{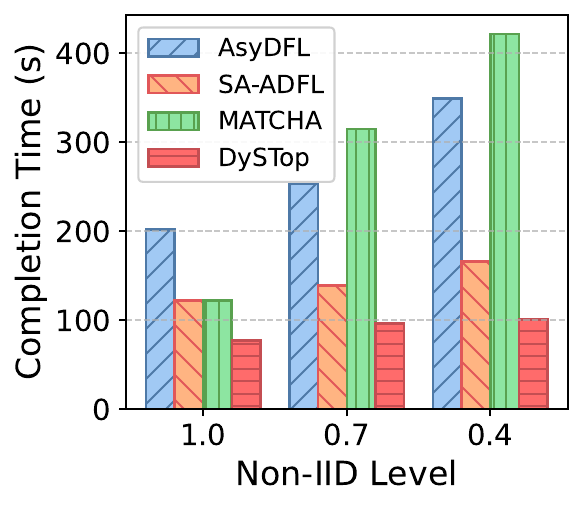}
\includegraphics[width=0.49\linewidth,height=3.6cm]{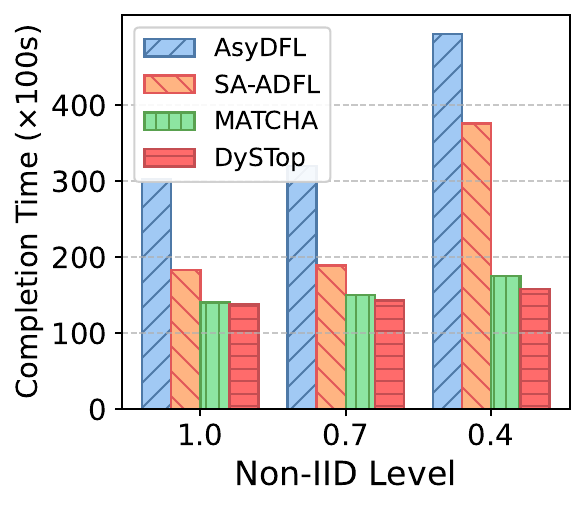}
\small{\caption{Completion time varies with different non-IID levels on the two datasets. \textit{Left}: FMNIST; \textit{Right}: CIFAR-10.}\label{fig:comp_time}}
\vspace{-5mm}
\end{minipage}%
\hspace{2mm}
\begin{minipage}[t]{0.49\linewidth}
\includegraphics[width=0.49\linewidth,height=3.6cm]{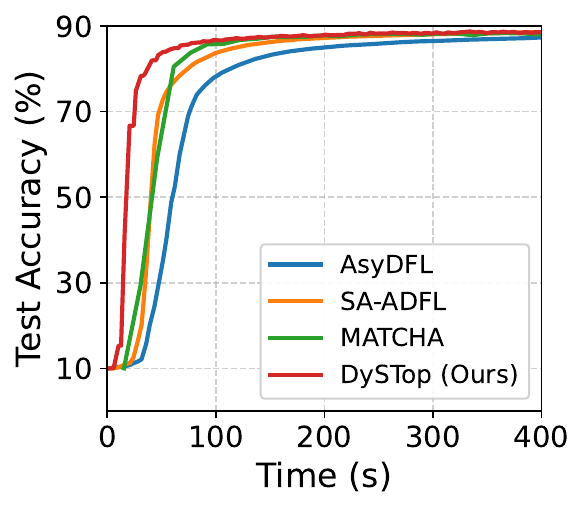}
\includegraphics[width=0.49\linewidth,height=3.6cm]{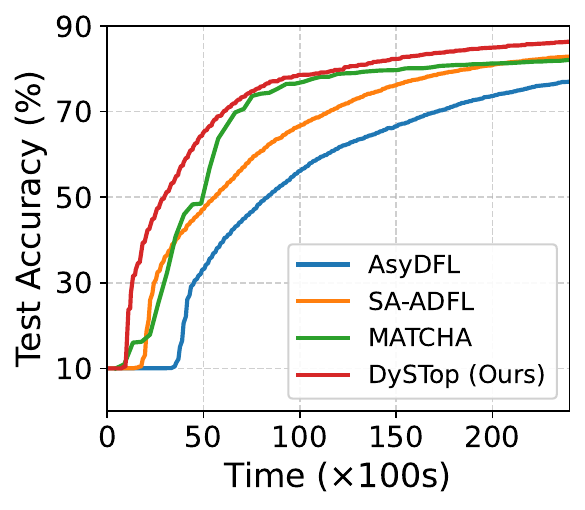}
\small{\caption{Test Accuracy vs. Time on the two datasets ($\phi = 1.0$). \textit{Left}: FMNIST; \textit{Right}: CIFAR-10.}\label{fig:acc_time_1.0}}
\vspace{-5mm}
\end{minipage}%
\end{figure*}

% a = 1.0
\begin{figure*}[t]
\begin{minipage}[t]{0.49\linewidth}
\includegraphics[width=0.49\linewidth,height=3.6cm]{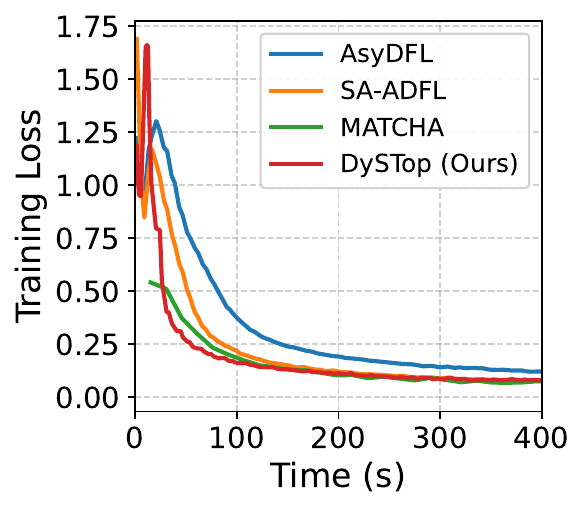}
\includegraphics[width=0.49\linewidth,height=3.6cm]{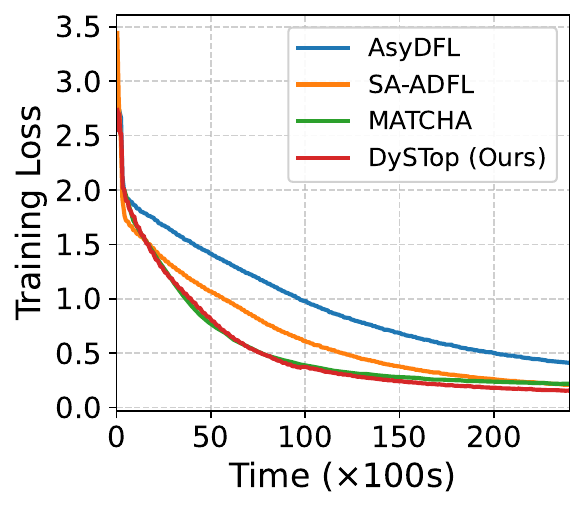}
\small{\caption{Training Loss vs. Time on the two datasets ($\phi = 1.0$). \textit{Left}: FMNIST; \textit{Right}: CIFAR-10.}\label{fig:loss_time_1.0}}
\vspace{-5mm}
\end{minipage}%
\hspace{2mm}
\begin{minipage}[t]{0.49\linewidth}
\includegraphics[width=0.49\linewidth,height=3.6cm]{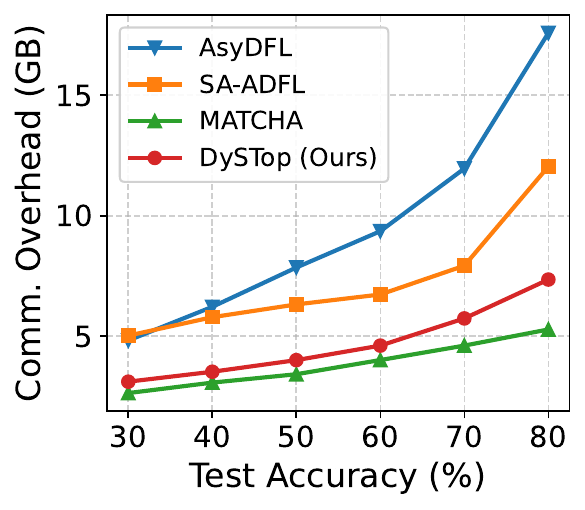}
\includegraphics[width=0.49\linewidth,height=3.6cm]{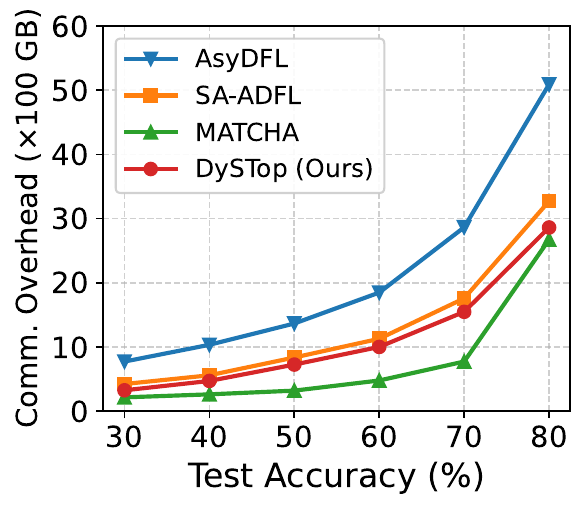}
\small{\caption{Communication Overhead vs. Test Accuracy on the two datasets ($\phi = 1.0$). \textit{Left}: FMNIST; \textit{Right}: CIFAR-10.}\label{fig:comm_acc_1.0}}
\vspace{-5mm}
\end{minipage}%
\end{figure*}
\subsubsection{Environment}\label{subsubsec:testbed}

We perform our simulations on a workstation with one Intel(R) Xeon(R) CPU and 4 NVIDIA GeForce RTX 3090Ti GPUs with 24GB of video memory. The system environment is CUDA v12.1, and cuDNN v8.9.2,
We use PyTorch to simulate a large-scale federated learning system with 100 heterogeneous workers distributed randomly across a 100m$\times$100m region. Specifically, the model training time is calculated by $H_i = \zeta_i\frac{D_i}{|\xi_i|}$, where $\zeta_i$ is the training time of one mini-batch on worker $v_i$. We measure the actual training time for one batch 100 times on our experimental equipment. We then multiply this by a coefficient drawn from a normal distribution to simulate the varying computational capabilities of heterogeneous devices. For model transmission, we use Shannon Formula to simulate the transmission rate $r_t^{i,j} = b \cdot \log_2\left(1+\frac{p_j\cdot g_t^{i,j}}{\gamma^2}\right)$. Referring to \cite{wang2023Decentralized,wang2022asyfed}, the channel gain $g_t^{i, j}$ is exponentially distributed with mean $G_0\cdot \mbox{Dist}(v_i, v_j)^{-4}$, where $G_0 = -43 \text{dB}$ is the path-loss constant at a reference distance of 1 m. The transmit power $p_j$ is set to range between 10dBm and 20dBm, then we multiply it by a fluctuation sampled from a normal distribution for each worker's transmit power. The noise power is set as $\gamma=10^{-13}$ W and the bandwidth is set as $b=1$MHz.

\subsubsection{Datasets and Models}\label{subsubsec:modelsanddatasets}
We select two classical benchmark datasets, FMNIST \cite{lecun1998gradient} and CIFAR-10 \cite{krizhevsky2009learning}, respectively. FMNIST consists of 60,000 grayscale images of fashion items for training and 10,000 for testing from 10 classes, while CIFAR-10 comprises 50,000 color images for training and 10,000 for testing depicting 10 classes, like animals and vehicles.
We also choose two classical models, CNN \cite{shalev2014understanding} and ResNet-18 \cite{he2016deep}, and customize them to adapt to the corresponding datasets.
The CNN is composed of two 5$\times$5 convolution layers, two 2$\times$2 max pooling layers, two fully-connected layers, and one softmax layer for output. The ResNet-18 is a deeper model, which is composed of 17 residual blocks and one fully-connected layer. We use the Dirichlet distribution\cite{lin2020ense} to generate the different levels $\phi=1.0, 0.7, 0.4$ of non-IID datasets. A smaller value of $\phi$ represents a more uneven distribution of datasets among workers. Specially, when $\phi=1.0$, data is IID among workers.

\subsubsection{Benchmarks}\label{subsubsec:metric}

\begin{figure*}[t]
\begin{minipage}[t]{0.49\linewidth}
\includegraphics[width=0.49\linewidth,height=3.6cm]{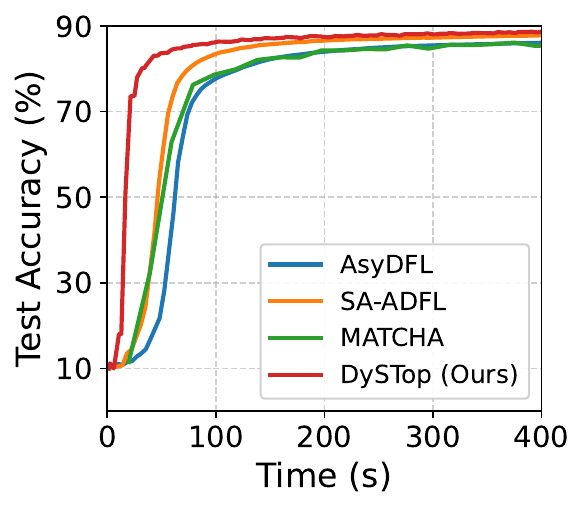}
\includegraphics[width=0.49\linewidth,height=3.6cm]{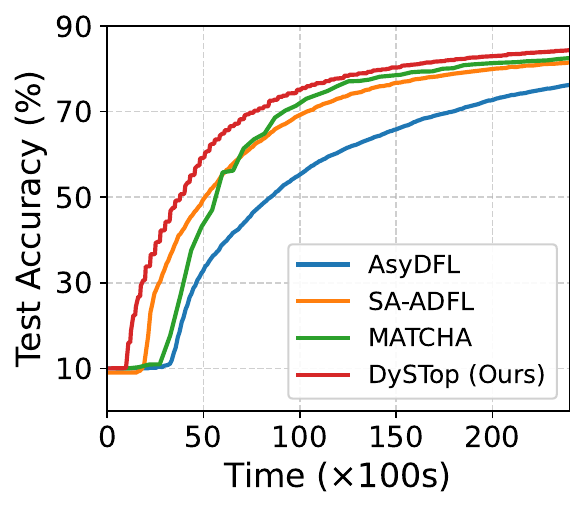}
\small{\caption{Test Accuracy vs. Time on the two datasets ($\phi = 0.7$). \textit{Left}: FMNIST; \textit{Right}: CIFAR-10.}\label{fig:acc_time_0.7}}
\vspace{-5mm}
\end{minipage}%
\hspace{2mm}
\begin{minipage}[t]{0.49\linewidth}
\includegraphics[width=0.49\linewidth,height=3.6cm]{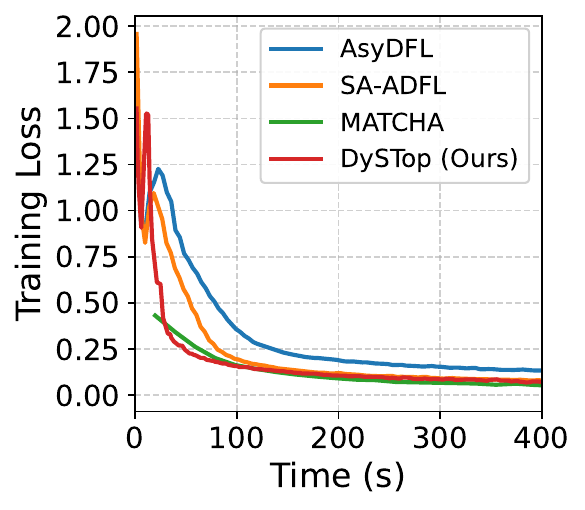}
\includegraphics[width=0.49\linewidth,height=3.6cm]{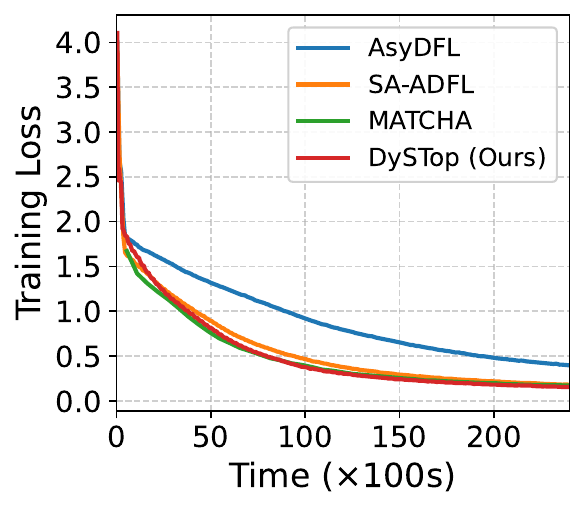}
\small{\caption{Training Loss vs. Time on the two datasets ($\phi = 0.7$). \textit{Left}: FMNIST; \textit{Right}: CIFAR-10.}\label{fig:loss_time_0.7}}
\vspace{-5mm}
\end{minipage}%
\end{figure*}

We compare our DySTop mechanism with the following DFL benchmarks.
\begin{itemize}
    \item MATCHA \cite{wang2019matcha}: A synchronous decentralized federated learning mechanism that divides the base topology into disjoint subgraphs and parallelizes inter-subgraphs communication. Considering the synchronous mechanism and sparse topology, we regard MATCHA as a lower bound of communication overhead.
    \item AsyDFL \cite{liao2024asynchronous}: An asynchronous decentralized federated learning mechanism that incorporates neighbor selection and model push to handle non-IID data and edge heterogeneity. However, it lacks staleness control techniques.
    \item SA-ADFL \cite{ma2024dynamic}: An asynchronous decentralized federated learning mechanism that dynamically controls staleness, where each worker pushes its model to all neighbors.
\end{itemize}

\subsubsection{Performance Metrics}

We employ four key metrics to measure the performance of our proposed algorithms and the benchmark. 1) \textit{Test Accuracy} is calculated as the ratio of correctly predicted data points by the model to the total number of data points. We use the average test accuracy across all workers' local models as our evaluation metric. 2) \textit{Training Loss} reflects the model training process and indicates whether convergence has been achieved. 3) \textit{Communication Overhead} is calculated as the total bandwidth consumed by all workers to reach a certain accuracy. 4) \textit{Completion time} is adopted to measure the time when achieving convergence.

\subsection{Evaluation Results}\label{subsec:results}

\subsubsection{Performance Comparison with Benchmarks}

%a = 0.7
\begin{figure*}[t]
\begin{minipage}[t]{0.49\linewidth}
\includegraphics[width=0.49\linewidth,height=3.6cm]{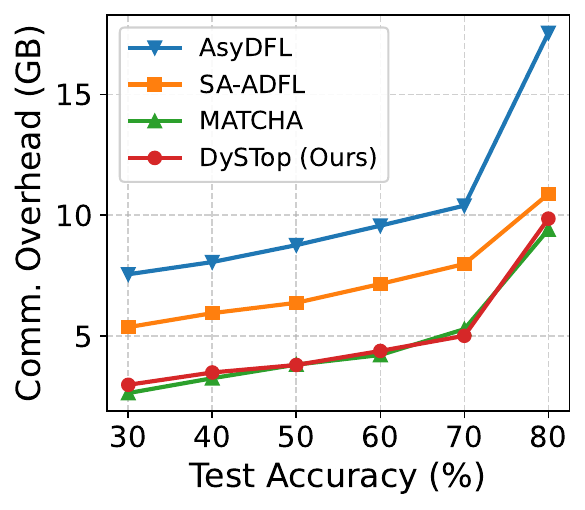}
\includegraphics[width=0.49\linewidth,height=3.6cm]{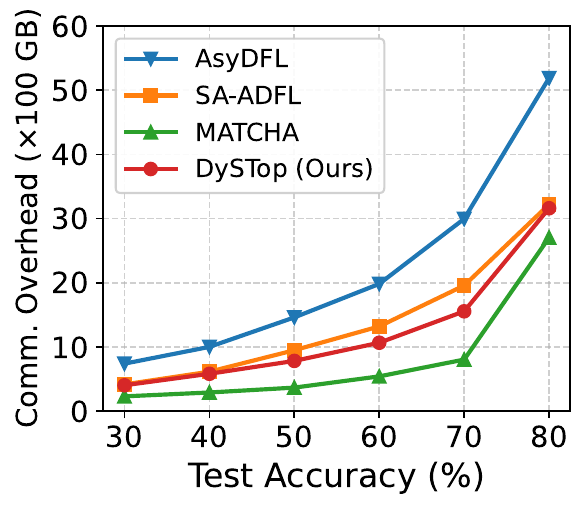}
\small{\caption{Communication Overhead vs. Test Accuracy on the two datasets ($\phi = 0.7$). \textit{Left}: FMNIST; \textit{Right}: CIFAR-10.}\label{fig:comm_acc_0.7}}
\vspace{-5mm}
\end{minipage}%
\hspace{2mm}
\begin{minipage}[t]{0.49\linewidth}
\includegraphics[width=0.49\linewidth,height=3.6cm]{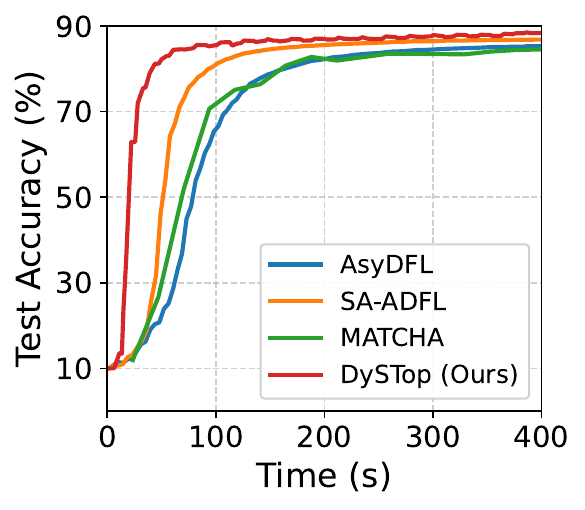}
\includegraphics[width=0.49\linewidth,height=3.6cm]{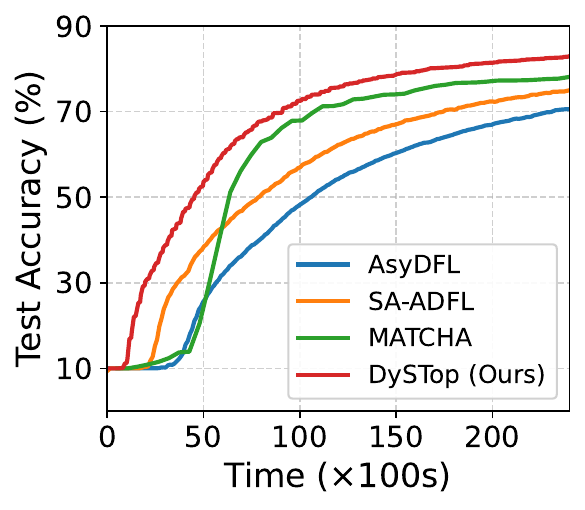}
\small{\caption{Test Accuracy vs. Time on the two datasets ($\phi = 0.4$). \textit{Left}: FMNIST; \textit{Right}: CIFAR-10.}\label{fig:acc_time_0.4}}
\vspace{-5mm}
\end{minipage}%
\end{figure*}

% a = 0.4
\begin{figure*}[t]
\begin{minipage}[t]{0.49\linewidth}
\includegraphics[width=0.49\linewidth,height=3.6cm]{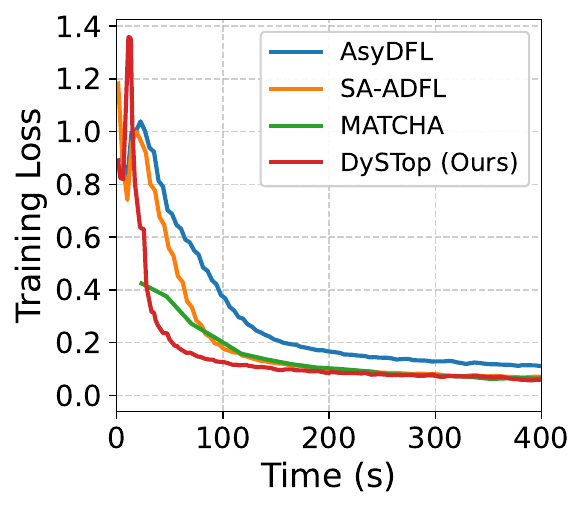}
\includegraphics[width=0.49\linewidth,height=3.6cm]{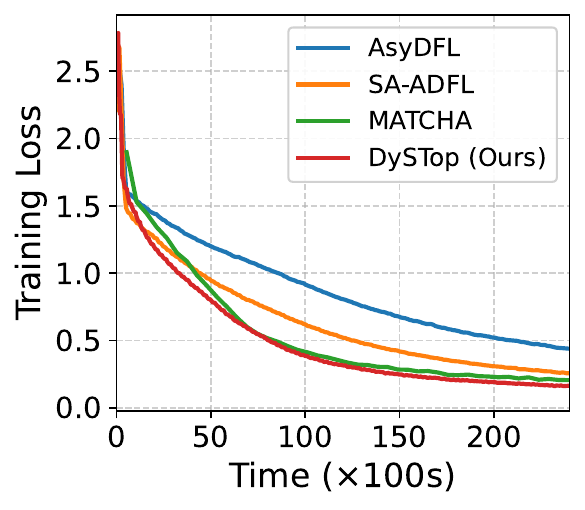}
\small{\caption{Training Loss vs. Time on the two datasets ($\phi = 0.4$). \textit{Left}: FMNIST; \textit{Right}: CIFAR-10.}\label{fig:loss_time_0.4}}
\vspace{-5mm}
\end{minipage}%
\hspace{2mm}
\begin{minipage}[t]{0.49\linewidth}
\includegraphics[width=0.49\linewidth,height=3.6cm]{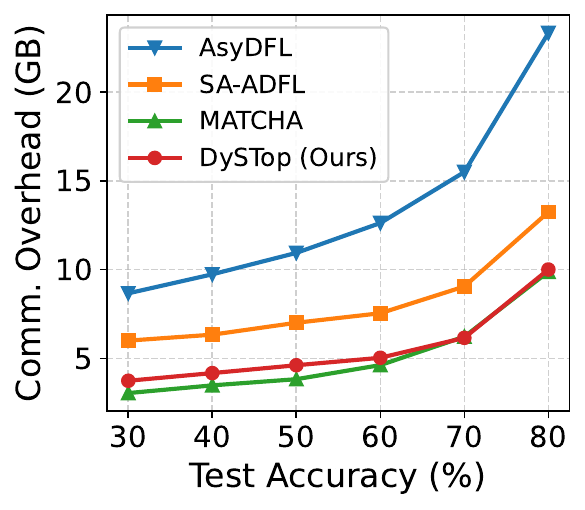}
\includegraphics[width=0.49\linewidth,height=3.6cm]{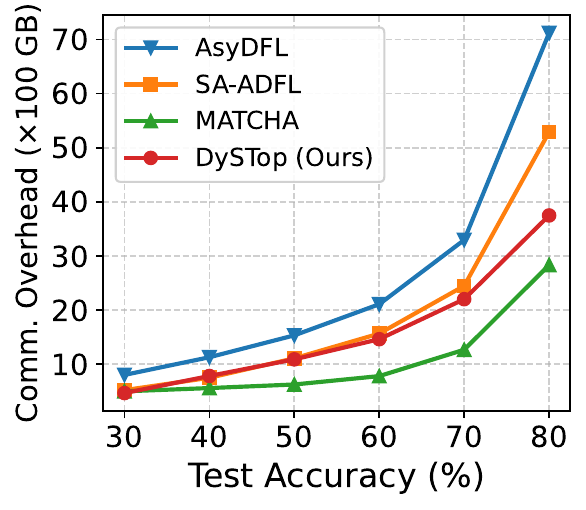}
\small{\caption{Communication Overhead vs. Test Accuracy on the two datasets ($\phi = 0.4$). \textit{Left}: FMNIST; \textit{Right}: CIFAR-10.}\label{fig:comm_acc_0.4}}
\vspace{-5mm}
\end{minipage}%
\end{figure*}

% fig 4
Fig. \ref{fig:comp_time} presents the completion time varying with different non-IID levels for training CNN on FMNIST and ResNet-18 on CIFAR-10. As shown in this figure, the completion time by the four algorithms on the two datasets increases with increasing non-IID level. For example, when the non-IID levels are $\phi=1.0$ and $\phi=0.4$, the completion time of DySTop takes 67.39s and 80.16s, increasing by 15.93\%, for training CNN on FMNIST, whereas AsyDFL, SA-ADFL, and MATCHA show larger increases of 42.15\%, 26.25\%, and 74.50\%. The experiment results demonstrate that DySTop can better deal with non-IID data, since DySTop dynamically constructs a topology considering non-IID data among workers. Meanwhile, DySTop can always take the least completion time in comparison with the other algorithms. For example, when the non-IID level is $\phi=0.4$, the completion times of DySTop, AsyDFL, SA-ADFL, and MATCHA are 80.16s, 349.27s, 166.35s and 422.76s, respectively, for training CNN on FMNIST. DySTop reduces the completion time by 77.04\%, 51.81\% and 81.03\% compared with AsyDFL, SA-ADFL, and MATCHA, respectively.

% fig 5-13
Figs. \ref{fig:acc_time_1.0}-\ref{fig:comm_acc_0.4} illustrate the performance of DySTop versus benchmarks in terms of test accuracy, training loss, and communication overhead for the values of $\phi$ are 1.0, 0.7, and 0.4, respectively.
Firstly, as the non-IID level increases, the performance of all algorithms decreases. For example, by Fig. \ref{fig:acc_time_1.0} and Fig. \ref{fig:acc_time_0.4}, for training ResNet-18 on CIFAR-10, the model accuracy of DySTop reaches 83.67\% in the IID scenario ($\phi=1.0$) but drops slightly to 81.36\% with a high Non-IID level ($\phi=0.4$) after 25000s.
Secondly, DySTop can achieve a faster convergence rate compared to AsyDFL, SA-ADFL, and MATCHA. For instance, as shown in Fig. \ref{fig:acc_time_0.4}, when taking 25000s for training ResNet-18 on CIFAR-10, DySTop achieves the test accuracy of 81.36\%, outperforming AsyDFL (67.18\%), SA-ADFL (72.39\%), and MATCHA (78.09\%).
Thirdly, DySTop consumes less communication overhead than asynchronous algorithms to achieve the required accuracy. For example, from Fig. \ref{fig:comm_acc_1.0}, when training CNN on FMNIST, to reach 80\% accuracy, the communication overhead of DySTop is 58.17\% and 38.94\% less than that of AsyDFL and SA-ADFL. This is because DySTop effectively constructs a sparse topology, reducing the communication resources consumption while ensuring accuracy.
Note that although DySTop consumes more communication overhead than MATCHA because of the more frequent communications of asynchronous schemes than that of synchronous schemes, DySTop can achieve higher test accuracy than MATCHA within the same training time.

\subsubsection{The Impact of Staleness Constraint}\label{subsubsec:stalness}

\begin{figure*}[t]
\begin{minipage}[t]{0.49\linewidth}
\includegraphics[width=0.49\linewidth,height=3.6cm]{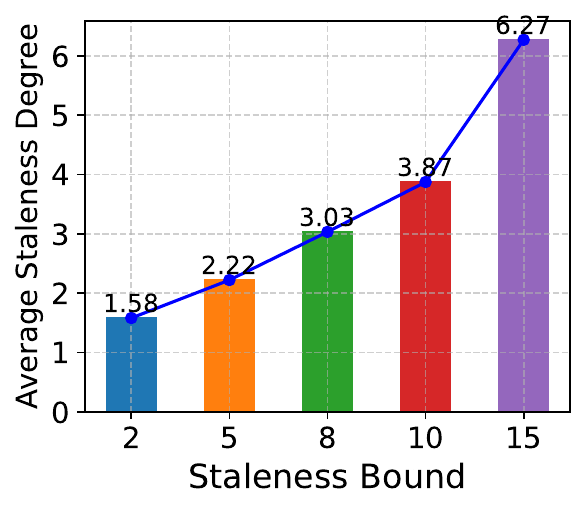}
\includegraphics[width=0.49\linewidth,height=3.6cm]{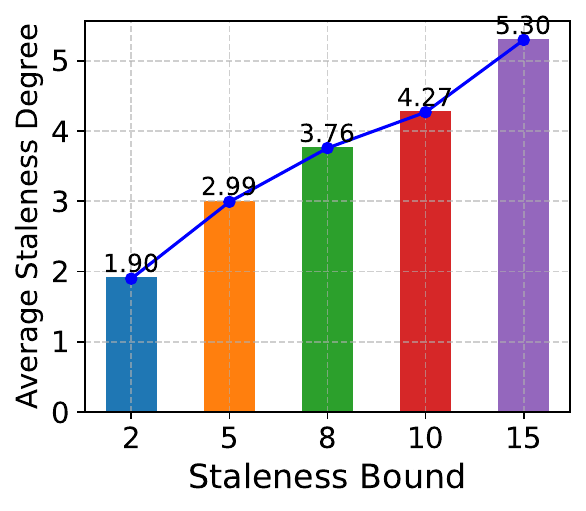}
\small{\caption{Average staleness degree varies with different settings of staleness bound in DySTop on the two datasets. \textit{Left}: FMNIST; \textit{Right}: CIFAR-10.}\label{fig:avg_stale}}
\vspace{-5mm}
\end{minipage}%
\hspace{2mm}
\begin{minipage}[t]{0.49\linewidth}
\includegraphics[width=0.49\linewidth,height=3.6cm]{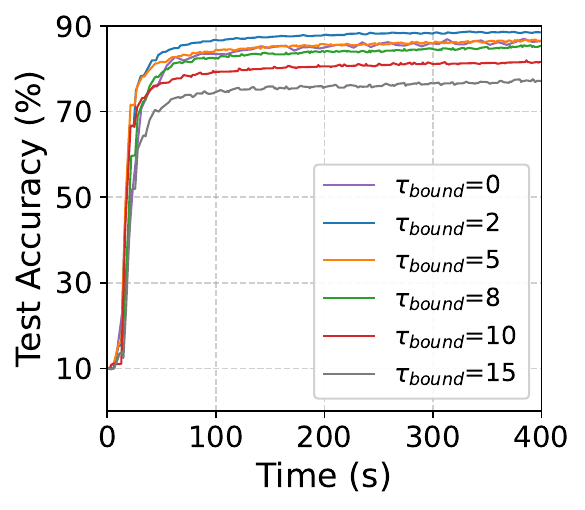}
\includegraphics[width=0.49\linewidth,height=3.6cm]{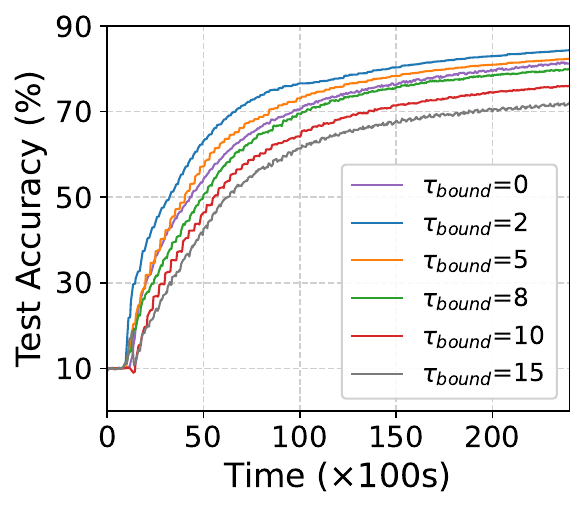}
\small{\caption{Test Accuracy vs. Time on the two datasets with various values of $\tau_{\rm bound}$. \textit{Left}: FMNIST; \textit{Right}: CIFAR-10.}\label{fig:acc_time_stale}}
\vspace{-5mm}
\end{minipage}%
\end{figure*}

\begin{figure*}[t]
\begin{minipage}[t]{0.49\linewidth}
\includegraphics[width=0.49\linewidth,height=3.6cm]{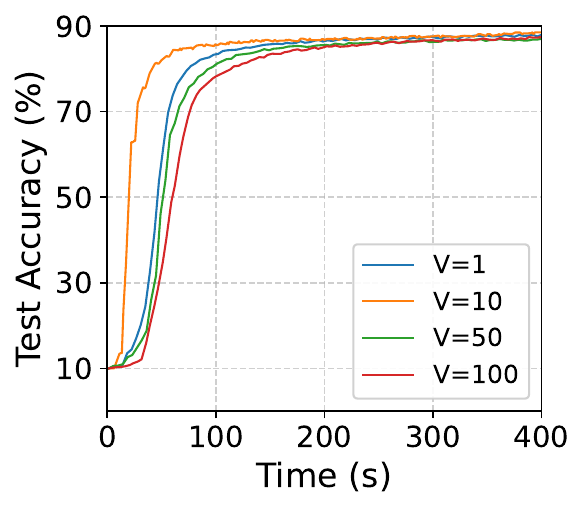}
\includegraphics[width=0.49\linewidth,height=3.6cm]{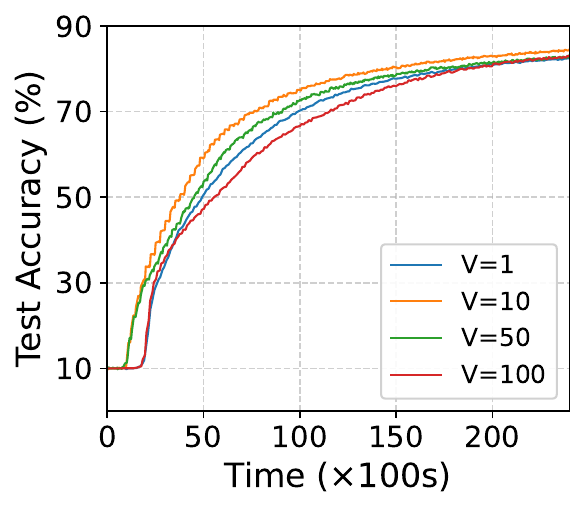}
\small{\caption{Test Accuracy vs. Time on the two datasets with various values of $V$. \textit{Left}: FMNIST; \textit{Right}: CIFAR-10.}\label{fig:acc_time_V}}
\vspace{-5mm}
\end{minipage}%
\hspace{2mm}
\begin{minipage}[t]{0.49\linewidth}
\includegraphics[width=0.49\linewidth,height=3.6cm]{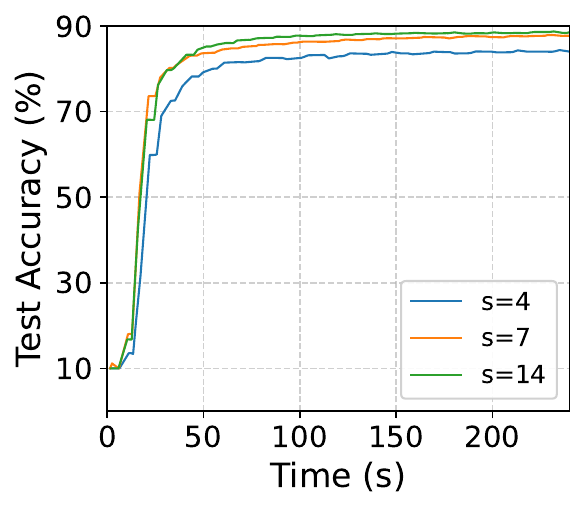}
\includegraphics[width=0.49\linewidth,height=3.6cm]{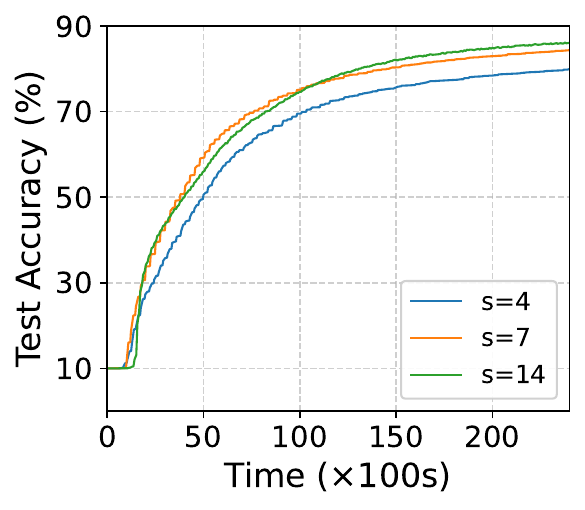}
\small{\caption{Test Accuracy vs. Time on the two datasets with various values of $s$. \textit{Left}: FMNIST; \textit{Right}: CIFAR-10.}\label{fig:acc_time_nei}}
\vspace{-5mm}
\end{minipage}%
\end{figure*}

% fig 12 13
Fig. \ref{fig:avg_stale} shows average staleness degree varying with different settings of staleness bound $\tau_{\rm bound}$ in DySTop on the two datasets. As shown in this figure, DySTop provides good control of staleness. That is, at the difference of staleness bounds of 2, 5, 8, 10, and 15, DySTop can control average staleness degrees as 1.58, 2.22, 3.03, 3.87, 6.27 on FMNIST and 1.90, 2.99, 3.76, 4.27, 5.30 on CIFAR-10. Fig. \ref{fig:acc_time_stale} illustrates the test accuracies with respect to the training time of various values of $\tau_{\rm bound} = \{0, 2, 5, 8, 10, 15\}$. Among these settings, $\tau_{\rm bound} = 2$ yields the highest accuracy, \eg 87.85\% when training a CNN on FMNIST. When staleness is too small (\eg, $\tau_{\rm bound} = 0$), the training process tends to become synchronous, leading to idle computation resources and reduced accuracy (\eg, 85.01\%). On the other hand, excessive staleness (\eg, $\tau_{\rm bound} = 15$) may cause excessive staleness and gradient divergence, which also decreases achieved accuracy (\eg, 75.61\%).

\subsubsection{The Impact of Control Parameter}\label{subsubsec:controlparameter}

% fig 14
Fig. \ref{fig:acc_time_V} illustrates the test accuracy with respect to the training time of the control parameter $V$. A larger value of $V$ places more importance on reducing training time, while a smaller value concentrates on staleness control.
Based on the results, we observe that overemphasizing the importance of either training time or staleness stability decreases the convergence rate.
For instance, after 25000s of training ResNet-18 on the CIFAR-10, the accuracy reaches to 82.67\% at $V=10$, compared to the accuracies of 81.01\%, 81.36\%, 80.98\% when $V=1, 50, 100$. Furthermore, the consumed training time of DySTop to achieve 70\% accuracy is 9961.43s, 7646.84s, 9277.89s, and 11545.63s when $V=1, 10, 50, 100$. Consequently, setting an appropriate value of $V$ can accelerate the model convergence, and we select $V=10$ for DySTop in the comparison experiments.

\subsubsection{The Impact of Neighbors Number}\label{subsubsec:performancecomparison}

\begin{figure}
\includegraphics[width=0.49\linewidth,height=3.6cm]{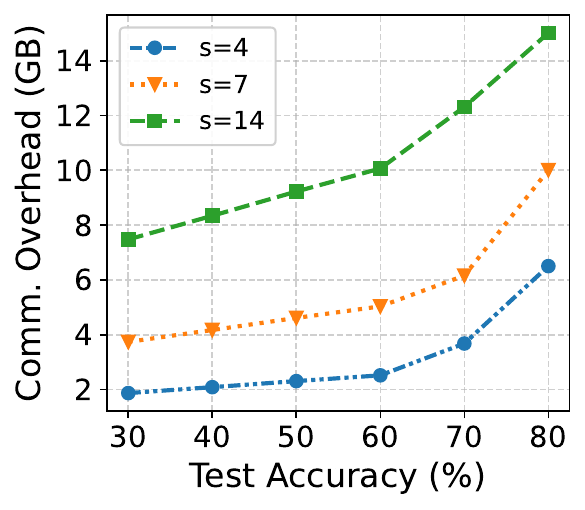}
\includegraphics[width=0.49\linewidth,height=3.6cm]{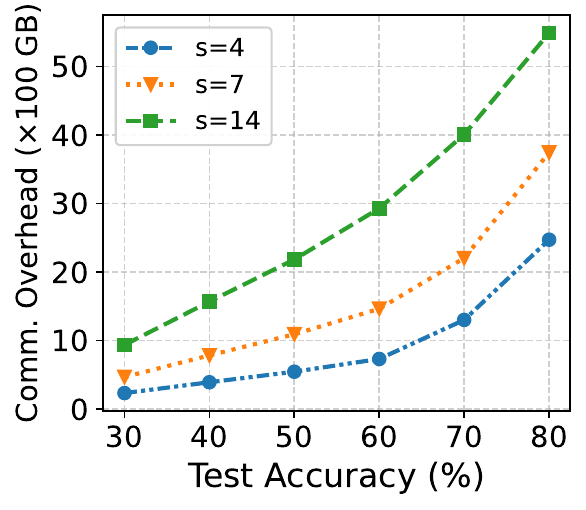}
\small{\caption{Communication overhead vs. Test Accuracy on the two datasets with various numbers of neighbors. \textit{Left}: FMNIST; \textit{Right}: CIFAR-10.}\label{fig:comm_acc_nei}}
\vspace{-5mm}
\end{figure}

To further investigate the impact of varying neighbor numbers on worker training performance.
We select different sample sizes $s$ as the number of in-neighbors of each worker. The values of $s$ are set as $\lceil \log_2N \rceil = 7$, $\lceil\frac{ \log_2N }{2}\rceil = 4$ and $\lceil 2\log_2N \rceil = 14$.
Fig. \ref{fig:acc_time_nei} shows the test accuracy with respect to the training time for varying values of sample size $s$. As shown in this figure, DySTop achieves model accuracies of 84.02\%, 87.60\%, and 88.35\% on FMNIST, and 78.29\%, 81.36\%, and 82.98\% on CIFAR-10, when $s=4,7,14$, respectively. This indicates that a larger $s$ value promotes the model convergence due to a denser network topology. However, this phenomenon exhibits diminishing returns, as the improvement in model accuracy slows with increasing values of $s$.

In Fig. \ref{fig:comm_acc_nei}, we compare the communication overhead required for different sample sizes $s$ to achieve a target test accuracy.
As shown in the figure, pulling models from a larger number of neighboring workers increases communication overhead. For example, to achieve 80\% accuracy,
DySTop consumes bandwidth of 6.50GB, 10.01GB, and 15.01GB for training a CNN on FMNIST, and 2472GB, 3743GB, and 5487GB for training ResNet-18 on CIFAR-10, when $s=4,7,14$, respectively. This indicates that a larger $s$ inevitably leads to a growth in communication overhead.
Figs. \ref{fig:acc_time_nei}-\ref{fig:comm_acc_nei} jointly illustrate that an appropriate number of neighbors can significantly reduce communication overhead while maintaining model accuracy.

\section{System Implementation}\label{sec:implementation}
In this section, we present a system implementation and deploy several algorithms on this real testbed platform.

\subsection{Testbed Setup}\label{subsec:testbed}

We implement a real testbed platform comprising one coordinator and 15 workers (labeled from $v_1$ to $v_{15}$). Fig. \ref{fig:testbed} illustrates a snapshot of the hardware devices in our testbed.

\begin{figure}[hbt]
	\centering
    \includegraphics[width=0.9\columnwidth]{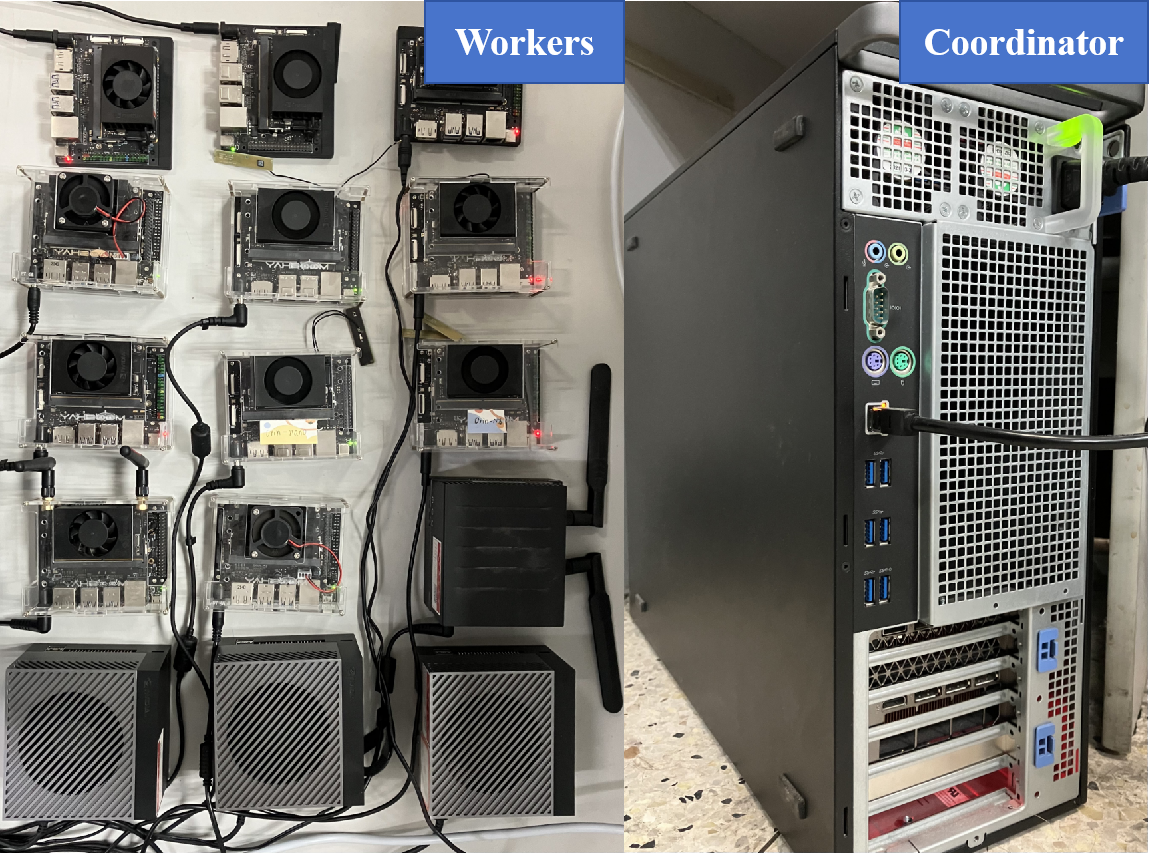}
 	\caption{
  A snapshot of the implemented testbed.}
	\label{fig:testbed}
\end{figure}

\begin{table}[h]
\centering
\small
\setlength{\tabcolsep}{0.5pt}
\caption{Real Testbed Platform Setting}
\label{tab:testbed}
\renewcommand{\arraystretch}{1.2}
\begin{tabular}{>{\bfseries}c|c|c|c}
\hline
\textbf{Type} & \textbf{Device} & \textbf{Component} & \textbf{Num.} \\
\hline
\multirow{5}{*}{Worker}
& Jetson Nano & \makecell[c]{128-core Maxwell GPU} & 4 \\
\cline{2-4}
& Jetson Orin Nano & \makecell[c]{1024-core Ampere GPU} & 3 \\
\cline{2-4}
& Jetson Orin NX & \makecell[c]{1024-core Ampere GPU} & 4 \\
\cline{2-4}
& Jetson Orin & \makecell[c]{2048-core Ampere GPU} & 3 \\
\cline{2-4}
& Jetson Xavier AGX  & \makecell[c]{512-core Volta GPU} & 1 \\
\hline
Coordinator & Dell T640 Server & Intel Xeon Silver 4210R CPU & 1 \\
\hline
\end{tabular}
\vspace{-2mm}
\end{table}

\begin{figure*}[t]
\begin{minipage}[t]{0.49\linewidth}
\includegraphics[width=0.49\linewidth,height=3.6cm]{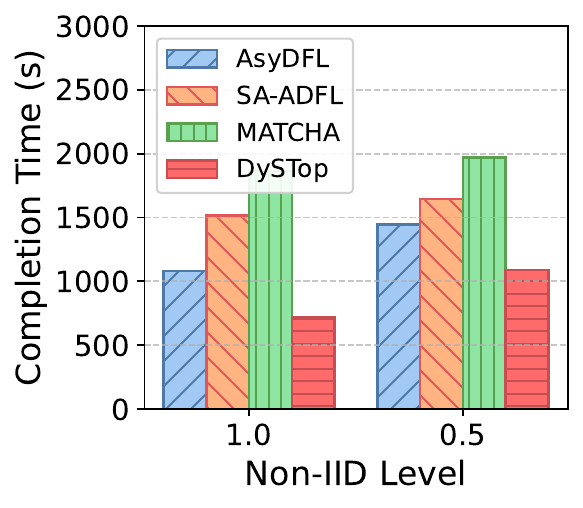}
\includegraphics[width=0.49\linewidth,height=3.6cm]{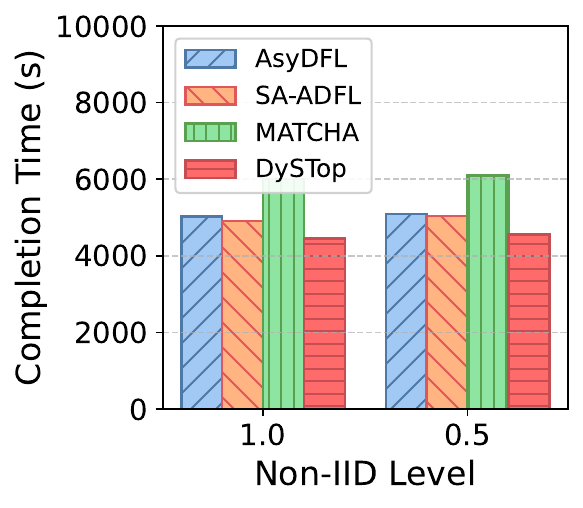}
\small{\caption{Completion time varies with different non-IID levels on the two datasets. \textit{Left}: SVHN; \textit{Right}: CIFAR-100.}\label{fig:comp_time_testbed}}
\vspace{-5mm}
\end{minipage}%
\hspace{2mm}
\begin{minipage}[t]{0.49\linewidth}
\includegraphics[width=0.49\linewidth,height=3.6cm]{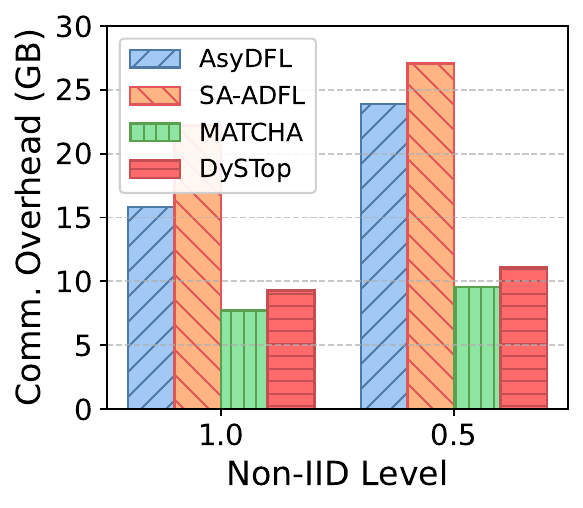}
\includegraphics[width=0.49\linewidth,height=3.6cm]{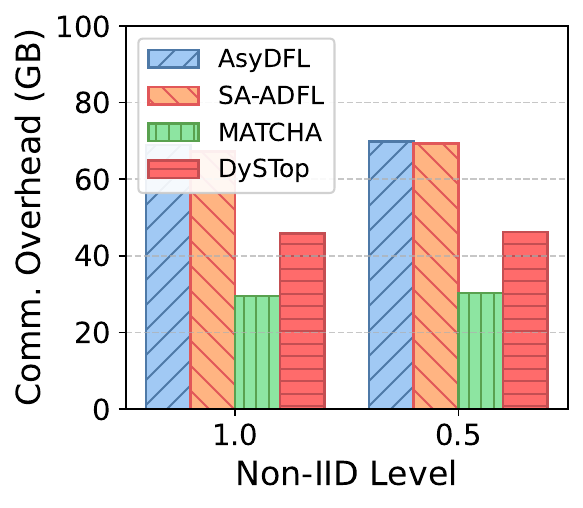}
\small{\caption{Communication overhead varies with different non-IID levels on the two datasets. \textit{Left}: SVHN; \textit{Right}: CIFAR-100.}\label{fig:comm_acc_testbed}}
\vspace{-5mm}
\end{minipage}%
\end{figure*}

\begin{figure*}[t]
\begin{minipage}[t]{0.49\linewidth}
\includegraphics[width=0.49\linewidth,height=3.6cm]{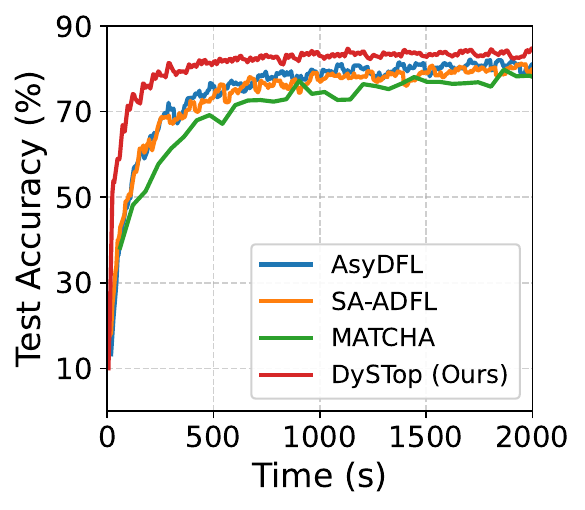}
\includegraphics[width=0.49\linewidth,height=3.6cm]{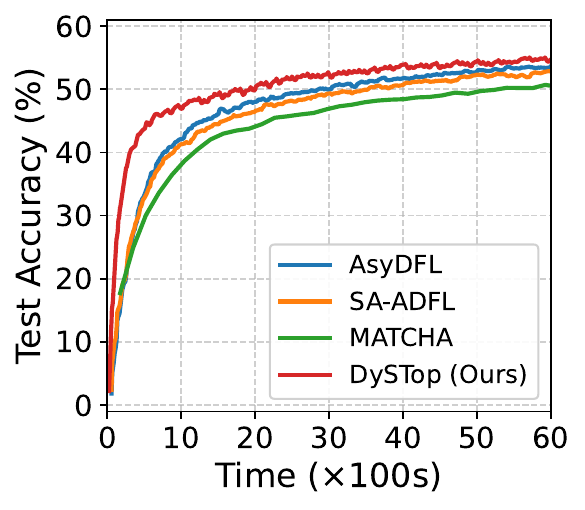}
\small{\caption{Test Accuracy vs. Time on the two datasets ($\phi = 1.0$). \textit{Left}: SVHN; \textit{Right}: CIFAR-100.}\label{fig:acc_time_1.0_testbed}}
\vspace{-5mm}
\end{minipage}%
\hspace{2mm}
\begin{minipage}[t]{0.49\linewidth}
\includegraphics[width=0.49\linewidth,height=3.6cm]{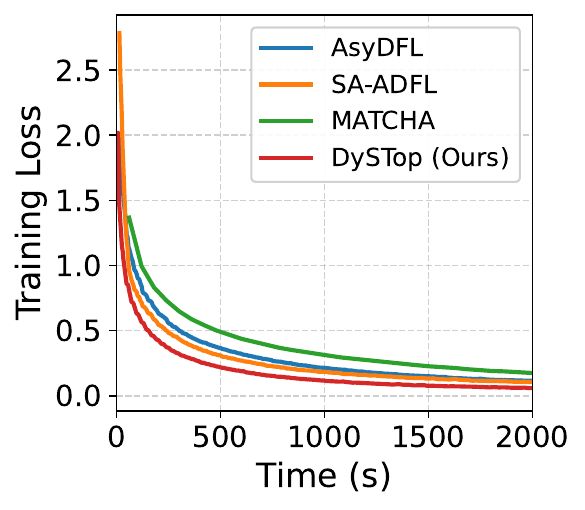}
\includegraphics[width=0.49\linewidth,height=3.6cm]{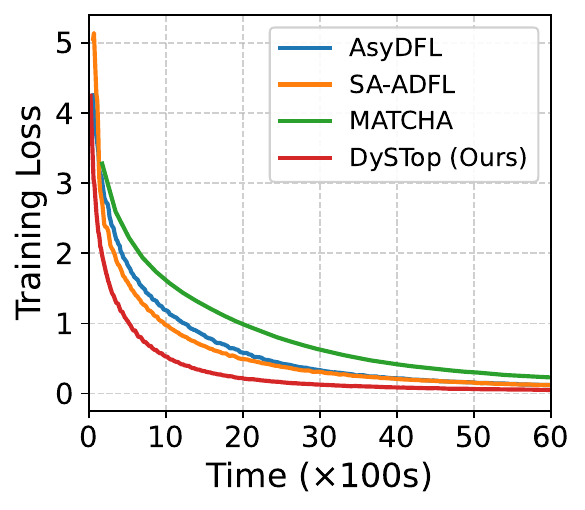}
\small{\caption{Training Loss vs. Time on the two datasets ($\phi = 1.0$). \textit{Left}: SVHN; \textit{Right}: CIFAR-100.}\label{fig:loss_time_1.0_testbed}}
\vspace{-5mm}
\end{minipage}%
\end{figure*}

Specifically, we set 4 Jetson Nano devices, 4 Jetson Orin NX devices, 3 Jetson Orin Nano devices, 3 Jetson Orin devices and 1 Jetson Xavier AGX device as heterogeneous workers. A Dell T640 Server is set as a coordinator. Detailed heterogeneous configurations are shown in Table \ref{tab:testbed}. Devices above are connected via wireless connection by a router, which is wired to the coordinator. Since the bandwidth in our experimental environment is much higher than in a real-world edge environment, we use the Wondershaper package to limit the bandwidth. The communication backend is based on mpi4py package, which is an MPI (Message Passing Interface) implementation for Python that enables multi-threaded communication across multiple nodes. We use mpi4py to achieve the model and information transmission among workers and coordinators.

The testbed experiments are conducted with two lightweight models designed for resource-constrained devices (\textit{i.e.}, SquezeeNet \cite{iandola2016squeezenet}, MobileNet-V2 \cite{howard2017mobilenets}) and two datasets (SVHN \cite{goodfellow2013multi} and CIFAR-100 \cite{krizhevsky2009learning}).
SquezeeNet is trained on SVHN, including 73,257 street view house numbers training images and 26,032 test images. To adapt SqueezeNet for the SVHN dataset, we adjust initial convolutional layer to have 96 output channels (with a 3$\times$3 kernel, stride 1, and padding 1) and replace the original classifier. The new classifier now comprises a dropout layer, a 1$\times$1 convolutional layer with 10 output channels (for the digit classes), a ReLU activation, and an adaptive average pooling layer to reduce feature maps to 1$\times$1. These changes allow SqueezeNet to effectively process SVHN's 3-channel images and classify them into 10 distinct digit categories. MobileNet-V2 is trained on CIFAR-100, consisting of 60,000 images for training and 10,000 for the test across 100 classes. For the MobileNetV2 model, we modify its initial convolutional layer and adjust the stride of certain feature blocks to 1 to preserve more spatial information. Additionally, we reconfigure the final convolutional layer within the feature module and replace the classifier's output layer to match the 100 CIFAR-100 classes. Similar to simulation experiments, we use the Dirichlet distribution to partition the datasets, setting the non-IID levels to 0.5 and 1.0.

Additionally, to verify the superiority of DySTop on our real testbed platform, we use the same benchmarks and metrics as those in Section \ref{sec:evaluation}.

\subsection{Testing Results}\label{subsec:testbedresults}

\begin{figure*}[t]
\begin{minipage}[t]{0.49\linewidth}
\includegraphics[width=0.49\linewidth,height=3.6cm]{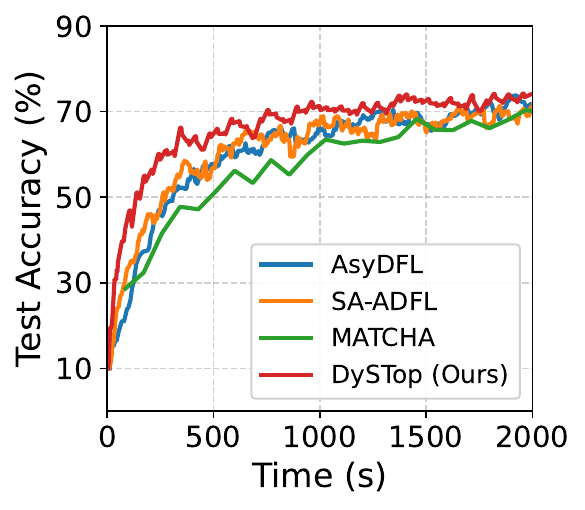}
\includegraphics[width=0.49\linewidth,height=3.6cm]{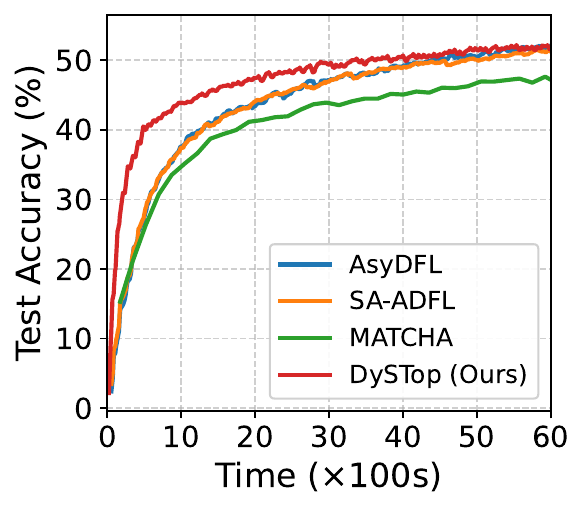}
\small{\caption{Test Accuracy vs. Time on the two datasets ($\alpha = 0.5$). \textit{Left}: SVHN; \textit{Right}: CIFAR-100.}\label{fig:acc_time_0.5_testbed}}
\vspace{-5mm}
\end{minipage}%
\hspace{2mm}
\begin{minipage}[t]{0.49\linewidth}
\includegraphics[width=0.49\linewidth,height=3.6cm]{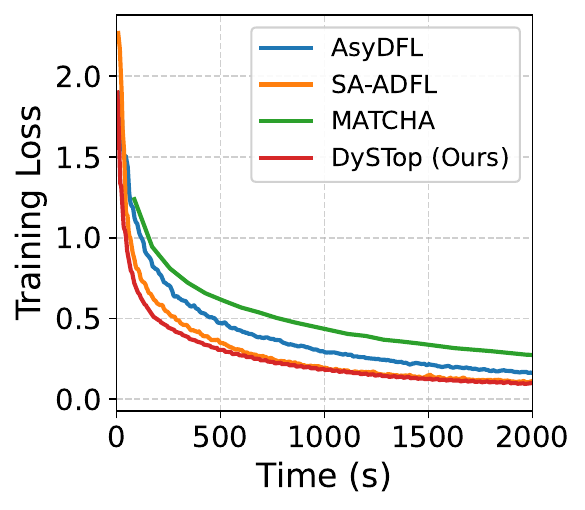}
\includegraphics[width=0.49\linewidth,height=3.6cm]{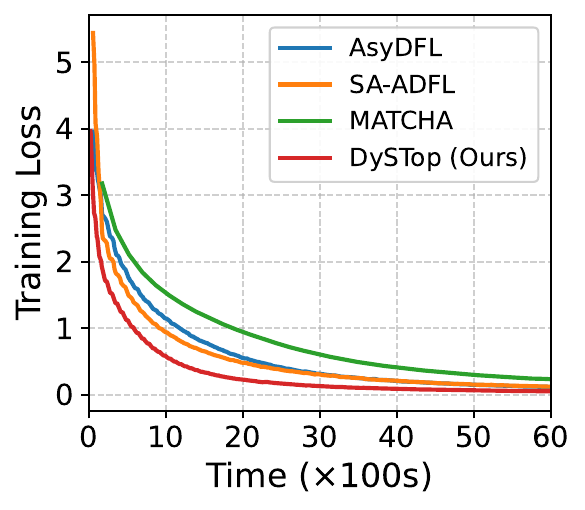}
\small{\caption{Training Loss vs. Time on the two datasets ($\phi = 0.5$). \textit{Left}: SVHN; \textit{Right}: CIFAR-100.}\label{fig:loss_time_0.5_testbed}}
\vspace{-5mm}
\vspace{-5mm}
\end{minipage}%
\end{figure*}

Fig. \ref{fig:comp_time_testbed} presents the completion time to achieve certain accuracies varying with different non-IID levels on the two datasets SVHN and CIFAR-100, respectively. As shown in the figure, DySTop can always take the least completion time compared with the other algorithms. For example, for training SqueezeNet on SVHN in IID scenario ($\phi = 1.0$), the completion time to achieve 80\% accuracy of AsyDFL, SA-ADFL, MATCHA and DySTop is 1079s, 1516s, 1865, and 721s, respectively. DySTop reduces the completion time by 33.17\%, 52.44\% and 61.34\% compared with AsyDFL, SA-ADFL and MATCHA.
In addition, when training SqueezeNet on SVHN in non-IID scenario ($\phi=0.5$), the completion time to achieve 70\% accuracy of AsyDFL, SA-ADFL, MATCHA and DySTop is 1445s, 1644s, 1969s and 1088s, respectively. DySTop reduces the completion time of AsyDFL, SA-ADFL and MATCHA by 24.71\%, 33.82\% and 44.73\%. The experiment results demonstrate that DySTop can better adapt to both the IID and non-IID scenarios with less completion time for model convergence.

Fig. \ref{fig:comm_acc_testbed} compares the communication overhead by different algorithms to achieve certain accuracies. As shown in the figure, the DySTop consumes the least communication overhead compared with other two asynchronous algorithms.
Furthermore, for training MobileNet-V2 on CIFAR-100 ($\phi=1.0$) to 50\% accuracy, the DySTop significantly reduces the communication overhead by 33.57\% and 31.82\%, compared with AsyDFL and SA-ADFL. When training MobileNet-V2 on CIFAR-100 ($\phi=0.5$), the communication overhead of DySTop is 33.92\% and 33.38\% less than that of AsyDFL and SA-ADFL.
From Figs. \ref{fig:comp_time_testbed} and \ref{fig:comm_acc_testbed}, although asynchronous DFL approaches inherently have a much higher communication frequency, DySTop still achieves faster model convergence with tolerable communication overhead.

Figs. \ref{fig:acc_time_1.0_testbed}-\ref{fig:loss_time_0.5_testbed} show training process for IID ($\phi=1.0$) and non-IID ($\phi=0.5$) scenarios, respectively.
DySTop can always achieve a faster convergence rate in both IID and non-IID settings, compared with AsyDFL, SA-ADFL and MATCHA.
For example, as shown in Fig. \ref{fig:acc_time_1.0_testbed}, when training SqueezeNet on SVHN with the IID setting, after 2000s of training, DySTop achieves a stable test accuracy of 84.49\%, outperforming AsyDFL (82.01\%), SA-ADFL (80.69\%) and MATCHA (79.37\%).
For the non-IID scenario shown in Fig. \ref{fig:acc_time_0.5_testbed},
after 6000s training in MobileNet-V2 on CIFAR-100, DySTop achieves a stable test accuracy of 52.18\%, outperforming AsyDFL (51.84\%), SA-ADFL (51.21\%) and MATCHA (46.61\%). The experiment results demonstrate that DySTop can still outperform under a real test-bed platform.

\begin{remark}
Compared to the simulation results, the relative performance of the mechanisms under the testbed experiment shows some differences. Specifically, MATCHA (synchronous DFL) performs worse in the testbed than in simulation, primarily due to the higher degree of device heterogeneity in the testbed. 
This is a result of having fewer workers in the testbed compared to the simulation environment. 
As a consequence, the straggler problem is exacerbated in the testbed, with slower devices causing delays in the model aggregation process, further reducing training efficiency.
\end{remark}

\section{Conclusion}\label{sec:conclusion}

In this paper, we have proposed DySTop, an innovative asynchronous decentralized federated learning (ADFL) mechanism designed to overcome the limitations of the existing synchronous DFL and asynchronous DFL mechanisms. In each round, DySTop dynamically activates multiple workers, each selecting a subset of its neighbors and pulling their models for aggregation. We have proved the convergence of DySTop, quantitatively revealing how key factors, such as the maximum staleness, worker activation frequency, and data distribution, affect the convergence bound of ADFL. Guided by the analysis, we have designed a worker activation algorithm to control staleness, and a phase-aware topology construction algorithm to reduce communication overhead and handle non-IID data. Extensive experiments on both simulation and testbeds have demonstrated the superiority of DySTop.

\bibliographystyle{IEEEtran}
\bibliography{contents/ref}

% Generated by IEEEtran.bst, version: 1.14 (2015/08/26)
\begin{thebibliography}{10}
\providecommand{\url}[1]{#1}
\csname url@samestyle\endcsname
\providecommand{\newblock}{\relax}
\providecommand{\bibinfo}[2]{#2}
\providecommand{\BIBentrySTDinterwordspacing}{\spaceskip=0pt\relax}
\providecommand{\BIBentryALTinterwordstretchfactor}{4}
\providecommand{\BIBentryALTinterwordspacing}{\spaceskip=\fontdimen2\font plus
\BIBentryALTinterwordstretchfactor\fontdimen3\font minus \fontdimen4\font\relax}
\providecommand{\BIBforeignlanguage}[2]{{%
\expandafter\ifx\csname l@#1\endcsname\relax
\typeout{** WARNING: IEEEtran.bst: No hyphenation pattern has been}%
\typeout{** loaded for the language `#1'. Using the pattern for}%
\typeout{** the default language instead.}%
\else
\language=\csname l@#1\endcsname
\fi
#2}}
\providecommand{\BIBdecl}{\relax}
\BIBdecl

\bibitem{hu2024industrial}
Y.~Hu, Q.~Jia, Y.~Yao, Y.~Lee, M.~Lee, C.~Wang, X.~Zhou, R.~Xie, and F.~R. Yu, ``{Industrial Internet of Things Intelligence Empowering Smart Manufacturing: A Literature Review},'' \emph{IEEE Internet of Things Journal}, 2024.

\bibitem{zhou2023swarm}
J.~Zhou, Y.~Shen, L.~Li, C.~Zhuo, and M.~Chen, ``{Swarm Intelligence-Based Task Scheduling for Enhancing Security for IoT Devices},'' \emph{IEEE Transactions on Computer-Aided Design of Integrated Circuits and Systems}, vol.~42, no.~6, pp. 1756--1769, 2023.

\bibitem{mcmahan2017communication}
B.~McMahan, E.~Moore, D.~Ramage, S.~Hampson, and B.~A. y~Arcas, ``{Communication-efficient Learning of Deep Networks from Decentralized Data},'' in \emph{Artificial intelligence and statistics}.\hskip 1em plus 0.5em minus 0.4em\relax PMLR, 2017, pp. 1273--1282.

\bibitem{wang2019adaptive}
S.~Wang, T.~Tuor, T.~Salonidis, K.~K. Leung, C.~Makaya, T.~He, and K.~Chan, ``{Adaptive Federated Learning in Resource Constrained Edge Computing Systems},'' \emph{IEEE Journal on Selected Areas in Communications}, vol.~37, no.~6, pp. 1205--1221, 2019.

\bibitem{Li2020On}
X.~Li, K.~Huang, W.~Yang, S.~Wang, and Z.~Zhang, ``{On the Convergence of FedAvg on Non-IID Data},'' in \emph{International Conference on Learning Representations}.\hskip 1em plus 0.5em minus 0.4em\relax JMLR, 2020.

\bibitem{roy2019braintorrent}
A.~G. Roy, S.~Siddiqui, S.~P{\"o}lsterl, N.~Navab, and C.~Wachinger, ``{Braintorrent: A Peer-to-peer Environment for Decentralized Federated Learning},'' \emph{arXiv preprint arXiv:1905.06731}, 2019.

\bibitem{tang2022gossipfl}
Z.~Tang, S.~Shi, B.~Li, and X.~Chu, ``{GossipFL: A Decentralized Federated Learning Framework with Sparsified and Adaptive Communication},'' \emph{IEEE Transactions on Parallel and Distributed Systems}, vol.~34, no.~3, pp. 909--922, 2022.

\bibitem{bellet2022d}
A.~Bellet, A.-M. Kermarrec, and E.~Lavoie, ``{D-cliques: Compensating for Data Heterogeneity with Topology in Decentralized Federated Learning},'' in \emph{2022 41st International Symposium on Reliable Distributed Systems (SRDS)}.\hskip 1em plus 0.5em minus 0.4em\relax IEEE, 2022, pp. 1--11.

\bibitem{wang2019matcha}
J.~Wang, A.~K. Sahu, Z.~Yang, G.~Joshi, and S.~Kar, ``{MATCHA: Speeding Up Decentralized SGD via Matching Decomposition Sampling},'' in \emph{2019 Sixth Indian Control Conference (ICC)}, 2019, pp. 299--300.

\bibitem{xu2021decentralized}
H.~Xu, M.~Chen, Z.~Meng, Y.~Xu, L.~Wang, and C.~Qiao, ``{Decentralized Machine Learning Through Experience-Driven Method in Edge Networks},'' \emph{IEEE Journal on Selected Areas in Communications}, vol.~40, no.~2, pp. 515--531, 2022.

\bibitem{liao2023adaptive}
Y.~Liao, Y.~Xu, H.~Xu, L.~Wang, and C.~Qian, ``{Adaptive Configuration for Heterogeneous Participants in Decentralized Federated Learning},'' in \emph{IEEE INFOCOM 2023-IEEE Conference on Computer Communications}.\hskip 1em plus 0.5em minus 0.4em\relax IEEE, 2023, pp. 1--10.

\bibitem{cao2021hadfl}
J.~Cao, Z.~Lian, W.~Liu, Z.~Zhu, and C.~Ji, ``{HADFL: Heterogeneity-aware Decentralized Federated Learning Framework},'' in \emph{2021 58th ACM/IEEE Design Automation Conference (DAC)}.\hskip 1em plus 0.5em minus 0.4em\relax IEEE, 2021, pp. 1--6.

\bibitem{chen2023enhancing}
M.~Chen, Y.~Xu, H.~Xu, and L.~Huang, ``{Enhancing Decentralized Federated Learning for Non-IID Data on Heterogeneous Devices},'' in \emph{2023 IEEE 39th International Conference on Data Engineering (ICDE)}.\hskip 1em plus 0.5em minus 0.4em\relax IEEE, 2023, pp. 2289--2302.

\bibitem{liao2024asynchronous}
Y.~Liao, Y.~Xu, H.~Xu, M.~Chen, L.~Wang, and C.~Qiao, ``{Asynchronous Decentralized Federated Learning For Heterogeneous Devices},'' \emph{IEEE/ACM Transactions on Networking}, 2024.

\bibitem{ma2024dynamic}
Q.~Ma, J.~Liu, Q.~Jia, X.~Zhou, Y.~Hu, and R.~Xie, ``{Dynamic Staleness Control for Asynchronous Federated Learning in Decentralized Topology},'' in \emph{International Conference on Wireless Artificial Intelligent Computing Systems and Applications}.\hskip 1em plus 0.5em minus 0.4em\relax Springer, 2024, pp. 99--117.

\bibitem{cui2023hiera}
Y.~Cui, K.~Cao, J.~Zhou, and T.~Wei, ``{Optimizing Training Efficiency and Cost of Hierarchical Federated Learning in Heterogeneous Mobile-Edge Cloud Computing},'' \emph{IEEE Transactions on Computer-Aided Design of Integrated Circuits and Systems}, vol.~42, no.~5, pp. 1518--1531, 2023.

\bibitem{fu2024dta}
L.~Fu, N.~Cheng, X.~Wang, R.~Sun, N.~Lu, Z.~Su, and C.~Li, ``{DTA-RL: Dynamic Topology Adaptive Reinforcement Learning Approach for Task Offloading in Mobile Edge Computing},'' in \emph{GLOBECOM 2024 - 2024 IEEE Global Communications Conference}.\hskip 1em plus 0.5em minus 0.4em\relax IEEE, 2024, pp. 3334--3339.

\bibitem{xie2019asynchronous}
C.~Xie, S.~Koyejo, and I.~Gupta, ``{Asynchronous Federated Optimization},'' \emph{arXiv preprint arXiv:1903.03934}, 2019.

\bibitem{ma2021fedsa}
Q.~Ma, Y.~Xu, H.~Xu, Z.~Jiang, L.~Huang, and H.~Huang, ``{FedSA: A Semi-Asynchronous Federated Learning Mechanism in Heterogeneous Edge Computing},'' \emph{IEEE Journal on Selected Areas in Communications}, vol.~39, no.~12, pp. 3654--3672, 2021.

\bibitem{palmieri2024impact}
L.~Palmieri, C.~Boldrini, L.~Valerio, A.~Passarella, and M.~Conti, ``{Impact of Network Topology on the Performance of Decentralized Federated Learning},'' \emph{Computer Networks}, vol. 253, p. 110681, 2024.

\bibitem{wu2020safa}
W.~Wu, L.~He, W.~Lin, R.~Mao, C.~Maple, and S.~A. Jarvis, ``{SAFA: A Semi-asynchronous Protocol for Fast Federated Learning with Low Overhead},'' \emph{IEEE Transactions on Computers}, 2020.

\bibitem{sun2024staleness}
S.~Sun, Z.~Zhang, Q.~Pan, M.~Liu, Y.~Wang, T.~He, Y.~Chen, and Z.~Wu, ``{Staleness-controlled Asynchronous Federated Learning: Accuracy and Efficiency Tradeoff},'' \emph{IEEE Transactions on Mobile Computing}, 2024.

\bibitem{chen2020asynchronous}
Y.~Chen, Y.~Ning, M.~Slawski, and H.~Rangwala, ``{Asynchronous Online Federated Learning For Edge Devices with Non-IID Data},'' in \emph{2020 IEEE International Conference on Big Data (Big Data)}.\hskip 1em plus 0.5em minus 0.4em\relax IEEE, 2020, pp. 15--24.

\bibitem{ma2025air}
Q.~Ma, J.~Zhou, X.~Hou, J.~Liu, H.~Xu, J.~Miao, and Q.~Jia, ``{Air-FedGA: A Grouping Asynchronous Federated Learning Mechanism Exploiting Over-The-Air Computation},'' in \emph{2025 IEEE International Parallel and Distributed Processing Symposium (IPDPS)}, 2025, pp. 1--12.

\bibitem{zheng2017asynchronous}
S.~Zheng, Q.~Meng, T.~Wang, W.~Chen, N.~Yu, Z.-M. Ma, and T.-Y. Liu, ``{Asynchronous Stochastic Gradient Descent with Delay Compensation},'' in \emph{International Conference on Machine Learning}.\hskip 1em plus 0.5em minus 0.4em\relax PMLR, 2017, pp. 4120--4129.

\bibitem{zhu2022client}
H.~Zhu, J.~Kuang, M.~Yang, and H.~Qian, ``{Client Selection With Staleness Compensation in Asynchronous Federated Learning},'' \emph{IEEE Transactions on Vehicular Technology}, vol.~72, no.~3, pp. 4124--4129, 2022.

\bibitem{J1986round}
J.~Tsitsiklis, D.~Bertsekas, and M.~Athans, ``{Distributed Asynchronous Deterministic and Stochastic Gradient Optimization Algorithms},'' \emph{IEEE Transactions on Automatic Control}, vol.~31, no.~9, pp. 803--812, 1986.

\bibitem{ma2024feduc}
Q.~Ma, Y.~Xu, H.~Xu, J.~Liu, and L.~Huang, ``{FedUC: A Unified Clustering Approach for Hierarchical Federated Learning},'' \emph{IEEE Transactions on Mobile Computing}, 2024.

\bibitem{nguyen2018sgd}
L.~Nguyen, P.~H. Nguyen, M.~Dijk, P.~Richt{\'a}rik, K.~Scheinberg, and M.~Tak{\'a}c, ``{SGD and Hogwild! Convergence without the Bounded Gradients Assumption},'' in \emph{International Conference on Machine Learning}.\hskip 1em plus 0.5em minus 0.4em\relax PMLR, 2018, pp. 3750--3758.

\bibitem{wang2025fedsiam}
X.~Wang, Y.~Wang, M.~Yang, F.~Li, X.~Wu, L.~Fan, and S.~He, ``{FedSiam-DA: Dual-Aggregated Federated Learning via Siamese Network for Non-IID Data},'' \emph{IEEE Transactions on Mobile Computing}, vol.~24, no.~2, pp. 985--998, 2025.

\bibitem{rubner2000earth}
Y.~Rubner, C.~Tomasi, and L.~J. Guibas, ``{The Earth Mover's Distance as a Metric For Image Retrieval},'' \emph{International journal of computer vision}, vol.~40, pp. 99--121, 2000.

\bibitem{de2023epidemic}
M.~De~Vos, S.~Farhadkhani, R.~Guerraoui, A.-M. Kermarrec, R.~Pires, and R.~Sharma, ``{Epidemic Learning: Boosting Decentralized Learning with Randomized Communication},'' \emph{Advances in Neural Information Processing Systems}, vol.~36, pp. 36\,132--36\,164, 2023.

\bibitem{wang2023Decentralized}
F.~Wang, S.~Cai, and V.~K.~N. Lau, ``{Decentralized DNN Task Partitioning and Offloading Control in MEC Systems With Energy Harvesting Devices},'' \emph{IEEE Journal of Selected Topics in Signal Processing}, vol.~17, no.~1, pp. 173--188, 2023.

\bibitem{wang2022asyfed}
Z.~Wang, Z.~Zhang, Y.~Tian, Q.~Yang, H.~Shan, W.~Wang, and T.~Q.~S. Quek, ``{Asynchronous Federated Learning Over Wireless Communication Networks},'' \emph{IEEE Transactions on Wireless Communications}, vol.~21, no.~9, pp. 6961--6978, 2022.

\bibitem{lecun1998gradient}
Y.~LeCun, L.~Bottou, Y.~Bengio, P.~Haffner \emph{et~al.}, ``{Gradient-based Learning Applied to Document Recognition},'' \emph{Proceedings of the IEEE}, vol.~86, no.~11, pp. 2278--2324, 1998.

\bibitem{krizhevsky2009learning}
A.~Krizhevsky, G.~Hinton \emph{et~al.}, \emph{{Learning Multiple Layers of Features from Tiny Images}}.\hskip 1em plus 0.5em minus 0.4em\relax Citeseer, 2009.

\bibitem{shalev2014understanding}
S.~Shalev-Shwartz and S.~Ben-David, \emph{{Understanding Machine Learning: From Theory to Algorithms}}.\hskip 1em plus 0.5em minus 0.4em\relax Cambridge university press, 2014.

\bibitem{he2016deep}
K.~He, X.~Zhang, S.~Ren, and J.~Sun, ``{Deep Residual Learning for Image Recognition},'' in \emph{Proceedings of the IEEE conference on computer vision and pattern recognition}, 2016, pp. 770--778.

\bibitem{lin2020ense}
T.~Lin, L.~Kong, S.~U. Stich, and M.~Jaggi, ``{Ensemble Distillation for Robust Model Fusion in Federated Learning},'' in \emph{Proceedings of the 34th International Conference on Neural Information Processing Systems}.\hskip 1em plus 0.5em minus 0.4em\relax Curran Associates Inc., 2020.

\bibitem{iandola2016squeezenet}
F.~N. Iandola, S.~Han, M.~W. Moskewicz, K.~Ashraf, W.~J. Dally, and K.~Keutzer, ``{SqueezeNet: AlexNet-level accuracy with 50x fewer parameters and< 0.5 MB model size},'' \emph{arXiv preprint arXiv:1602.07360}, 2016.

\bibitem{howard2017mobilenets}
A.~G. Howard, M.~Zhu, B.~Chen, D.~Kalenichenko, W.~Wang, T.~Weyand, M.~Andreetto, and H.~Adam, ``{MobileNets: Efficient Convolutional Neural Networks for Mobile Vision Applications},'' \emph{arXiv preprint arXiv:1704.04861}, 2017.

\bibitem{goodfellow2013multi}
I.~J. Goodfellow, Y.~Bulatov, J.~Ibarz, S.~Arnoud, and V.~Shet, ``{Multi-digit Number Recognition from Street View Imagery Using Deep Convolutional Neural Networks},'' \emph{arXiv preprint arXiv:1312.6082}, 2013.

\end{thebibliography}

\vspace{-9mm}
\begin{IEEEbiography}[{\includegraphics[width=1in,height=1.25in,clip,keepaspectratio]{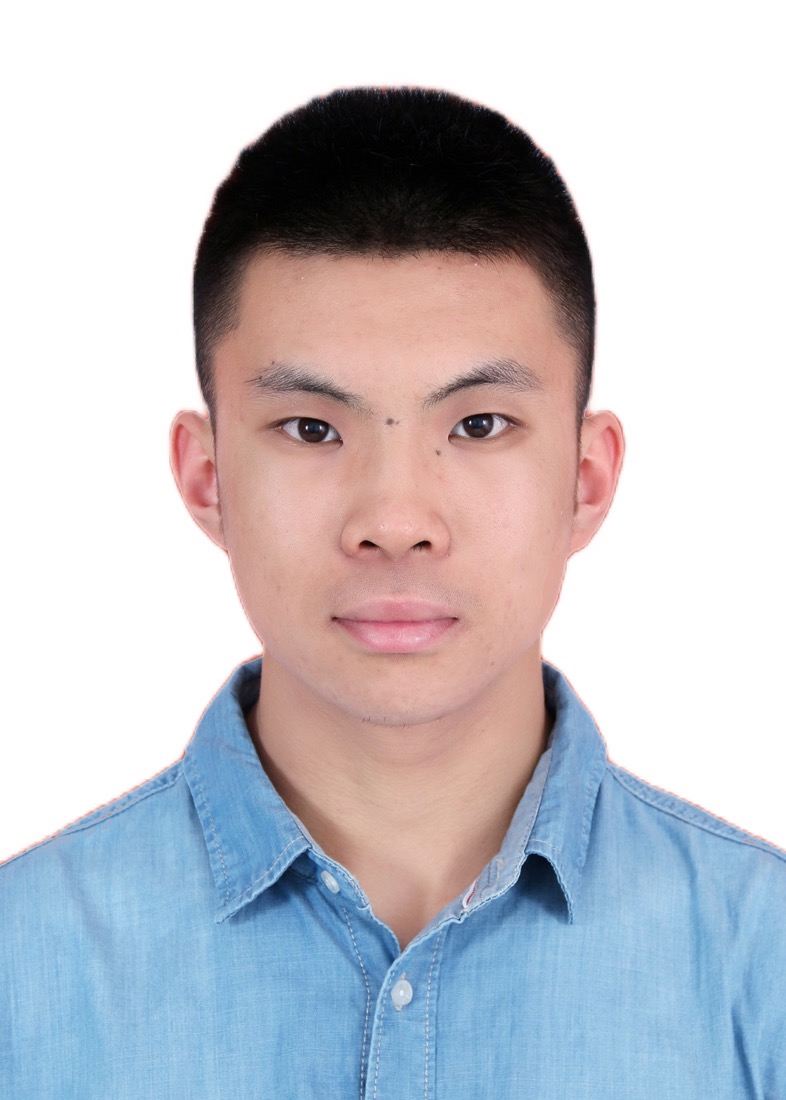}}]{Yizhou Shi} received the B.S. degree in Cyberspace Security from the Nanjing
University of Science and Technology,
Nanjing, China, in 2024. He is currently pursuing the Ph.D. degree with the Nanjing University of Science and Technology, Nanjing, China.
His research interests are in the area of edge computing and
federated learning.
\end{IEEEbiography}
\vspace{-10mm}
\begin{IEEEbiography}[{\includegraphics[width=1in,height=1.25in,clip,keepaspectratio]{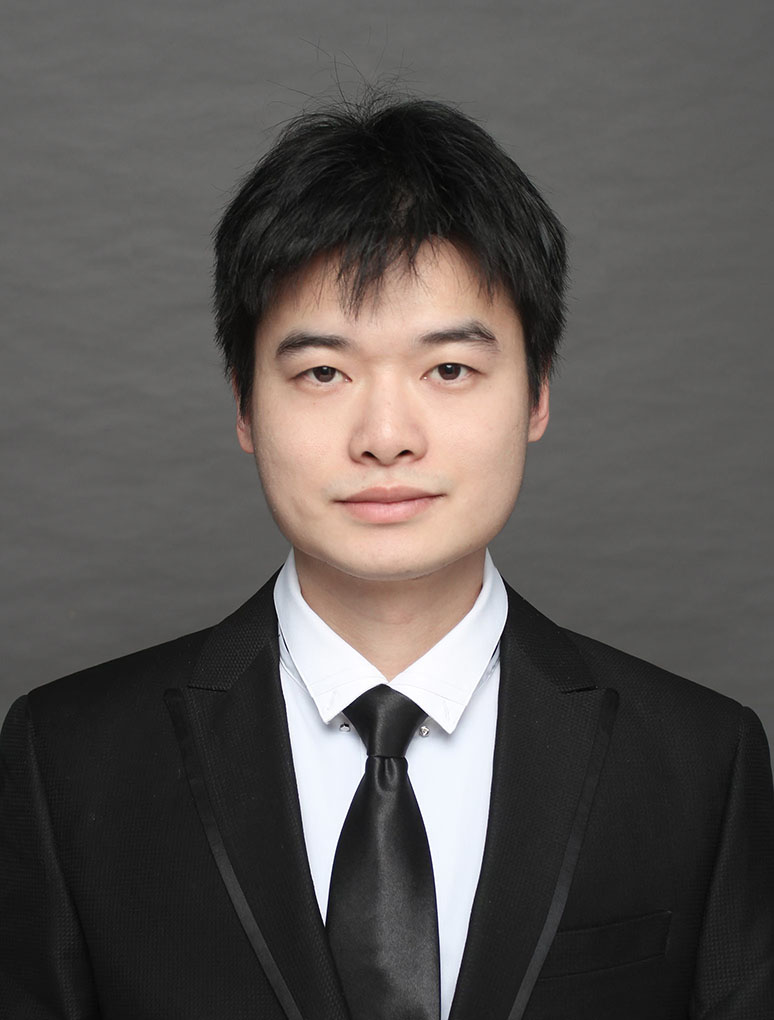}}]{Qianpiao Ma} received the B.S. degree in computer science and the Ph.D. degree in computer software and theory from the University of Science and Technology of China, Hefei, China, in 2014 and 2022, respectively. He is currently an Associate Professor at the School of Computer Science and Engineering, Nanjing University of Science and Technology, Nanjing, China. He has published more than 20 papers in famous journals and conferences, including IEEE JSAC, IEEE TMC, INFOCOM and IPDPS. His primary research interests include federated learning, mobile-edge computing, and distributed machine learning.
\end{IEEEbiography}
\vspace{-10mm}
\begin{IEEEbiography}[{\includegraphics[width=1in,height=1.25in,clip,keepaspectratio]{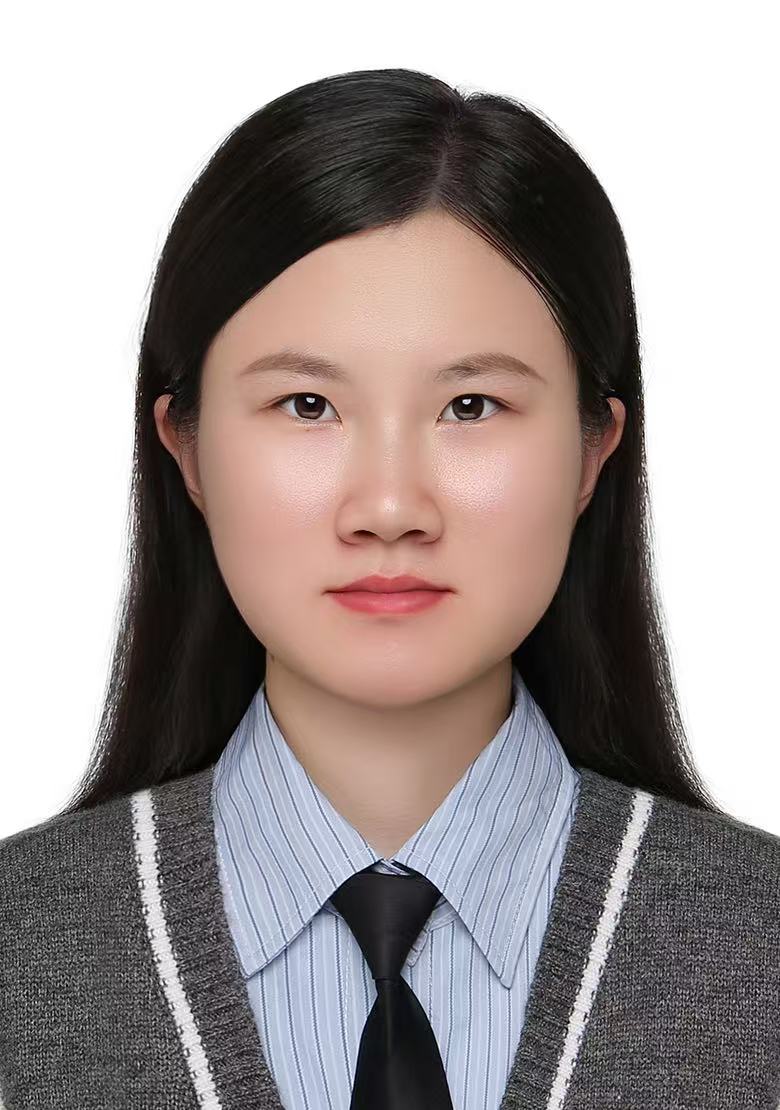}}]{Yan Xu} received the B.S. degree in Computer Science and Technology from the Nanjing University of Chinese Medicine, Nanjing, China, in 2025. She is currently pursuing the Ph.D. degree with the Nanjing University of Science and Technology, Nanjing, China. Her research interests are in the area of edge computing and federated learning.
\end{IEEEbiography}
\vspace{-10mm}
\begin{IEEEbiography}[{\includegraphics[width=1in,height=1.25in,clip,keepaspectratio]{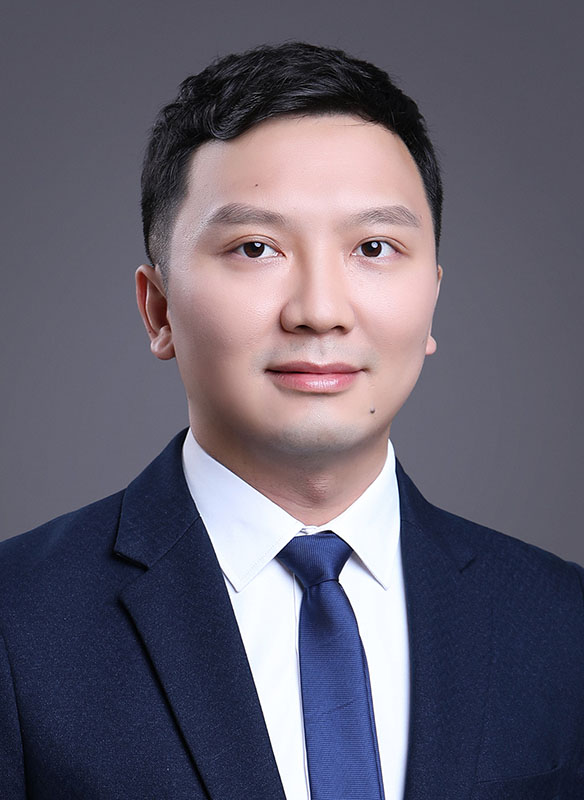}}]{Junlong Zhou} is an Associate Professor at the School of Computer Science and Engineering, Nanjing University of Science and Technology, China.
His research interests include edge computing, cloud computing, and embedded systems, where he has published 120 papers, including more than 40 in premier IEEE/ACM Transactions.
He has been a Subject Area Editor for Journal of Systems Architecture and an Associate Editor for Journal of Circuits, Systems, and Computers and IET Cyber-Physical Systems: Theory \& Applications.
He received the Best Paper Award from IEEE iThings2020, CPSCom2022, and ICITES2024.
\end{IEEEbiography}

\vspace{-10mm}
\begin{IEEEbiography}[{\includegraphics[width=1in,height=1.25in,clip,keepaspectratio]{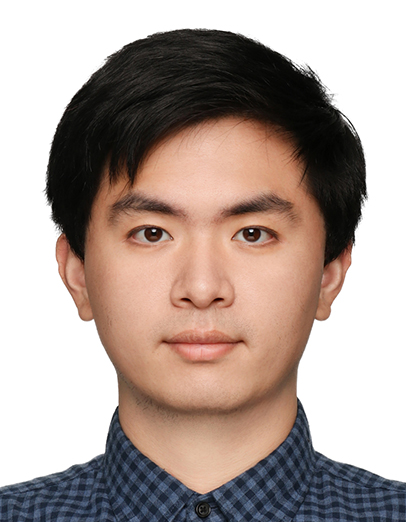}}]{Ming Hu} received his B.E. and Ph.D. degrees from the Software Engineering Institute, East China Normal University, Shanghai, China, in 2017 and 2022, respectively. He is currently a research scientist at Singapore Management University (SMU). Previously, he was a research fellow at Nanyang Technological University (NTU), Singapore. His research interests include the area of design automation of cyber-physical systems, federated learning, trustworthy AI, and software engineering.
\end{IEEEbiography}
\vspace{-10mm}
\begin{IEEEbiography}[{\includegraphics[width=1in,height=1.25in,clip,keepaspectratio]{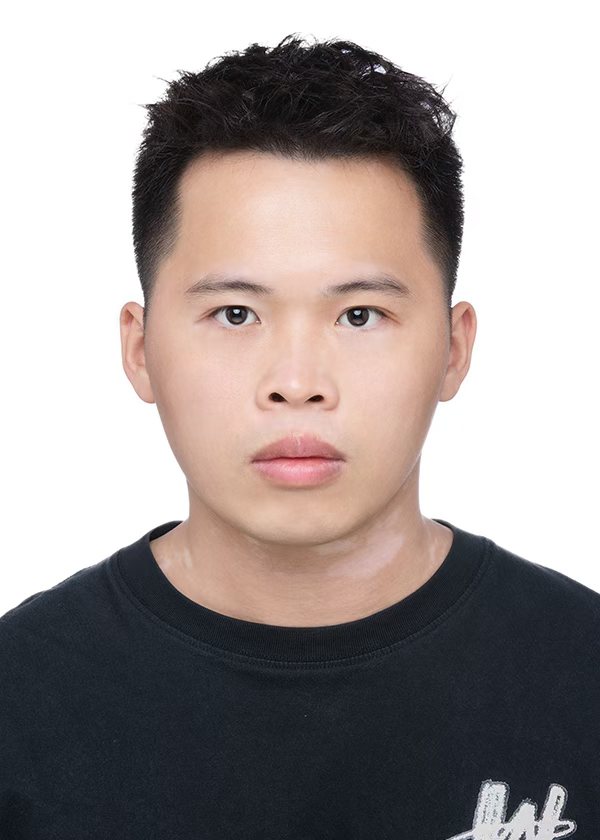}}] {Yunming Liao} received B.S. degree in 2020 from  the University of Science and Technology of China. He is currently pursuing his Ph.D. degree in the School of Computer Science and Technology, University of Science and Technology of China. His research interests include mobile edge computing, federated learning, and edge intelligence.
\end{IEEEbiography}
\vspace{-10mm}
\begin{IEEEbiography}[{\includegraphics[width=1in,height=1.25in,clip,keepaspectratio]{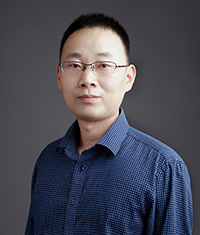}}]{Hongli Xu}(Member, IEEE) received the B.S. degree in computer science and the Ph.D. degree in computer software and theory from the University of Science and Technology of China, China, in 2002 and 2007, respectively. He is currently a Professor with the School of Computer Science and Technology, University of Science and Technology of China. He has published more than 100 papers in famous journals and conferences, including IEEE/ACM TNET, IEEE TMC, IEEE TPDS, INFOCOM, and ICNP. He has held more than 30 patents. His main research interests include software-defined networks, edge computing, and the Internet of Things. He was awarded the Outstanding Youth Science Foundation of NSFC in 2018. He has won the best paper award or the best paper candidate in several famous conferences.
\end{IEEEbiography}

\end{document}